\documentclass[%
 reprint,
 superscriptaddress,
 amsmath,amssymb,
 aps,
 prx,  
 floatfix,
]{revtex4-2}

\usepackage{graphicx}  
\usepackage{dcolumn}   
\usepackage{bm}        
\usepackage[colorlinks=true, allcolors=blue]{hyperref} 
\usepackage{xcolor}    
\usepackage{amsthm}
\usepackage[linesnumbered,algoruled,boxed,lined]{algorithm2e}

\newtheorem{theorem}{Theorem}

\newtheorem{proposition}{Proposition}

\newtheorem{lemma}{Lemma}
\newtheorem{corollary}{Corollary}

\DeclareMathOperator{\sign}{sign}
\DeclareMathOperator{\Tr}{Tr}
\DeclareMathOperator{\Var}{Var}

\begin{document}

\title{Stratified Sampling for Quasi-Probability Decompositions}
\author{Joshua W. Dai}
\email{joshua.dai@materials.ox.ac.uk}
\affiliation{Department of Materials, University of Oxford, Parks Road, OX2 6UD, United Kingdom}
\affiliation{Mathematical Institute, University of Oxford, Woodstock Road, Oxford OX2 6GG, United Kingdom}
\author{B\'alint Koczor}
\email{balint.koczor@maths.ox.ac.uk}
\affiliation{Mathematical Institute, University of Oxford, Woodstock Road, Oxford OX2 6GG, United Kingdom}

\date{\today}
\begin{abstract}
    Quasi-probability decompositions (QPDs) have proven essential in many quantum algorithms and protocols---one replaces a ``difficult'' quantum circuit with an ensemble of ``easier'' circuit variants whose weighted outcomes reproduce any target observable. This, however, inevitably yields an increased configuration variance beyond Born-rule shot noise. We develop a broad framework for accounting for and reducing this variance and prove that stratified sampling---under ideal proportional allocation---results in an unbiased estimator with a variance that is never worse than naïve sampling (with equality only in degenerate cases).
    Furthermore, we provide a classical dynamic programme to enable stratification on arbitrary product-form QPDs. Numerical simulations of typical QPDs, such as Probabilistic Error Cancellation (PEC) and Probabilistic Angle Interpolation (PAI), demonstrate constant-factor reductions in overall variance (up to $\sim 60$--$80\%$ in an oracle model) and robust $\sim 10\%$ savings in the pessimistic single-shot regime. Our results can be applied immediately to reduce the net sampling cost of practically relevant QPDs that are commonly used in near term and early fault-tolerant algorithms without requiring additional quantum resources.
\end{abstract}
\maketitle
\section{Introduction}
\label{sec:introduction}

Classical randomisation has been used to improve the effectiveness  of many quantum algorithms and protocols,
particularly in the context of near-term and early fault-tolerant quantum devices~\cite{preskill_beyond_2025, katabarwa_early_2024, zimboras_myths_2025}. Rather than repeatedly executing a single fixed quantum process, one samples a family of circuits or channels and combines outcomes in post-processing. This paradigm appears in quasi-probability decomposition (QPD) methods for simulation \cite{pashayan_estimating_2015, bennink_unbiased_2017, howard_application_2017, seddon_quantifying_2021}, error mitigation \cite{temme_error_2017, endo_practical_2018, bravyi_mitigating_2021, strikis_learning-based_2021, takagi_universal_2023, piveteau_error_2021, cai_quantum_2023, van_den_berg_probabilistic_2023, luthra_unlocking_2025, jeon_quantum_2026}, probabilistic synthesis \cite{koczor_probabilistic_2024, koczor_sparse_2024} and circuit cutting/knitting \cite{peng_simulating_2020, mitarai_constructing_2021, lowe_fast_2023}. Similar algorithmic randomisation ideas also form the basis for techniques such as randomised compiling and twirling \cite{wallman_noise_2016, hashim_randomized_2021, cai_constructing_2019}, randomised Hamiltonian simulation (e.g.\ qDRIFT) \cite{childs_faster_2019, campbell_random_2019, chen_concentration_2021, faehrmann_randomizing_2022, kiumi_te-pai_2024}, and randomised measurement schemes such as classical shadows \cite{huang_predicting_2020, huang_efficient_2021, jnane_quantum_2024, elben_randomized_2023}. While these techniques are often associated with NISQ applications \cite{preskill_quantum_2018}, randomisation can yield substantial resource savings in early fault-tolerant applications and beyond, through ideas such as statistical phase estimation \cite{wan_randomized_2022, kshirsagar_proving_2024, gunther_phase_2025, vu_low_2024}, linear algebra solvers via randomised linear combination of unitaries (LCU) \cite{wang_qubit-efficient_2024}, and randomised QSVT \cite{wang_randomized_2025} or randomised QSP \cite{martyn_halving_2025}. 

While QPDs can reduce quantum resource requirements or mitigate noise, they do introduce an additional variance source beyond intrinsic Born-rule shot noise: \emph{configuration variance} due to randomisation over quantum circuits/channels. At fixed target precision, this extra variance translates into an overhead in the number of distinct circuit variants (and hence total executions) required. For QPD-based protocols, the dominant scaling of this overhead is well known to be exponential in circuit size through the QPD $1$-norm $\|g\|_1$ \cite{takagi_fundamental_2022, zimboras_myths_2025}. Even when this exponential factor is tolerable, reducing constant prefactors in the variance can materially extend the regime where QPD-based methods are practically useful.

In this work we develop a sampling-design viewpoint on configuration variance and import variance-reduction tools from classical statistics. We model a hybrid protocol as sampling from a joint distribution over channels and outcomes, apply the law of total variance to separate Born-rule and configuration contributions, and then use conditioning sampling (Rao--Blackwellisation) to reduce the classical term without changing the target expectation \cite{kroese_handbook_2011}. For concreteness, we specialise to product-form QPDs: they are widely used, their randomisation structure is explicit, and they provide a clean setting in which algorithmic stratification can be made fully constructive. 

Our main technical contribution is to adapt \emph{stratified sampling} \cite{neyman_two_1992, cochran_sampling_1977, lohr_sampling_2021} to product-form QPD estimators and make it algorithmically explicit.
Specifically, we:
(i) prove that proportional allocation on any statistic is unbiased and (under ideal proportional quotas) never worse than naïve Monte Carlo at fixed configuration budget, with strict improvement except in degenerate cases;
(ii) construct a universal, index-based stratification of the configuration space using a \emph{counts vector} of local QPD indices;
(iii) for counts-vector stratification, give a dynamic-programming procedure that computes all stratum weights and supports exact conditional sampling, with pre-processing time $O(d\,\nu^{d})$ and $O(\nu^d)$ memory for local width $d$ and depth $\nu$ and per-sample conditional generation cost that matches naïve up to constant-factor $d$ (or identical assuming $O(d)$ categorical draws). 

Stratification changes the classical sampling plan and requires only classical pre- and post-processing and a conditional sampler, and does not need additional quantum resources beyond a naïve implementation. A schematic is shown in Fig.~\ref{fig:schematic}. We then verified our approach on first-order Trotter TFIM circuits with probabilistic angle interpolation (PAI) and probabilistic error cancellation (PEC), observing $\sim 10\%$ variance reductions in the worst-case shot noise dominant regimes and substantially larger reductions of variance (up to $\sim 60$--$80\%$) in oracle-like regimes where shot noise is suppressed.

\paragraph*{Related work.}
Conceptually, our approach is orthogonal to work that optimises the QPD \emph{design} itself (e.g.\ to reduce $\|g\|_1$): given any fixed product-form QPD, we address the sampling-design question of how to draw configurations and aggregate their outcomes as efficiently as possible while preserving unbiasedness. We note that there is substantial work on reducing QPD overheads by designing more efficient decompositions and task-specific mitigation schemes \cite{jiang_physical_2021, guo_quantum_2022, piveteau_quasiprobability_2022, scheiber_reduced_2025, tran_locality_2023, eddins_lightcone_2024, zhao_retrieving_2024}. Our contribution is complementary: we focus on unbiased sampling design and classical pre-processing for a fixed product-form QPD. A detailed comparison and connections to previous work on classical variance-reduction techniques in the quantum context, such as \cite{shyamsundar_cv4quantum_2025, chen_faster_2025}, are given in Appendix~\ref{app:related-work}. 

\paragraph*{Outline.}
Section~\ref{sec:QPDs} develops the variance decomposition for product-form QPD estimators, introduces stratified sampling and proportional/Neyman allocation, and states the ``never worse than naïve'' guarantee. We then construct counts-vector stratification and an efficient conditional sampler. Section~\ref{sec:numerical results} presents numerical results on PAI and PEC for TFIM Trotter evolution, including single-shot and oracle regimes. Section~\ref{sec:outlook} concludes and outlines extensions, while proofs, implementation details, and extended discussion appear in the appendices.

\begin{figure*}
    \centering
    \includegraphics[width=\textwidth]{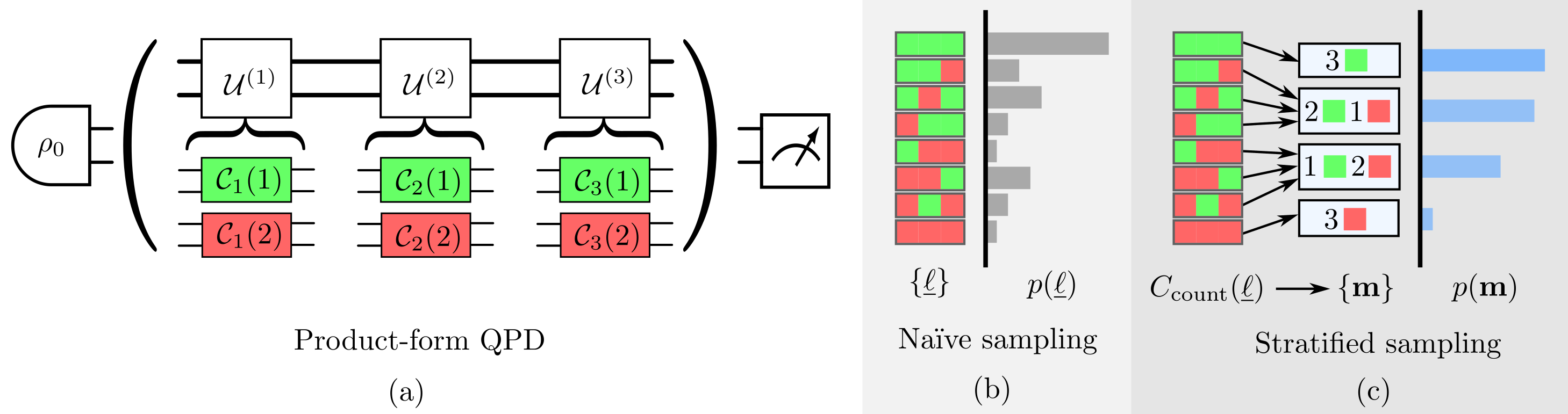}
    \caption{%
Schematic of product-form QPD sampling and counts-vector stratification on a toy three-gate circuit.
(a) Each ideal gate admits a local QPD over implementable primitives, so a circuit configuration is an index string $\underline\ell$.
(b) Naïve QPD sampling draws $\underline\ell$ directly from $p(\underline\ell)\propto |g(\underline\ell)|$.
(c) Counts-vector stratification groups configurations by the number of occurrences of each local primitive, yielding strata indexed by a counts vector $\mathbf M$ and a marginal distribution over strata obtained by aggregating the configuration masses.
}
    \label{fig:schematic}
\end{figure*}

\section{Product-form QPDs: variance decomposition and stratified sampling}
\label{sec:QPDs}

In this work we specialise to product-form quasi-probability decompositions (QPDs), encompassing probabilistic error cancellation, probabilistic synthesis, probabilistic angle interpolation (PAI), and related schemes. The input state $\rho_0$, measurement POVM $\{M_x\}$, and observable $O$ are fixed; all design randomness lies in sampling a circuit-level channel from a product ensemble of local QPDs. A more abstract treatment in terms of general channel randomisation is deferred to Appendix~\ref{app:general-channel-rand}.

\subsection{Local and circuit-level QPDs}
\label{subsec:local-circuit-qpd}
Suppose we have a target quantum process built from an ordered composition of \(\nu\) building blocks
\[
\mathcal U_{\mathrm{circ}}
= \mathcal U^{(\nu)}\circ\cdots\circ\mathcal U^{(1)},
\]
where each \(\mathcal U^{(i)}\) is a Hermiticity-preserving linear map on density operators \footnote{In many applications \(\mathcal U^{(i)}\) is CPTP (e.g.\ probabilistic synthesis or angle interpolation), but in others such as PEC it can be a formal inverse-noise map and hence not completely positive (see Appendix~\ref{app:qpd-instances}). The only operations we demand to be physical are the implementable primitives \(\mathcal C_i(\ell_i)\).}.

For each position \(i\in\{1,\dots,\nu\}\) we assume a \emph{local} quasi-probability decomposition
\begin{equation}
\mathcal U^{(i)}
= \sum_{\ell_i\in I_i}\gamma_i(\ell_i)\,\mathcal C_i(\ell_i),
\label{eq:local-QPD-main}
\end{equation}
where \(I_i=\{1,\dots,d_i\}\) indexes \emph{implementable} CPTP channels \(\mathcal C_i(\ell_i)\) and real coefficients \(\gamma_i(\ell_i)\in\mathbb R\) (possibly negative). Both \(\gamma_i\) and \(\mathcal C_i\) may vary with \(i\).

Although \(\mathcal U_{\mathrm{circ}}\) may fail to be CPTP when some \(\mathcal U^{(i)}\) are inverse maps (as in PEC), the estimator remains well-defined because each realised circuit \(\mathcal U(\underline\ell)\) is a composition of CPTP primitives.

Define the local 1-norm as $\|\boldsymbol\gamma_i\|_1 := \sum_{\ell_i\in I_i}|\gamma_i(\ell_i)|$ and the associated sampling distribution as $p_i(\ell_i) := |\gamma_i(\ell_i)|/\|\boldsymbol\gamma_i\|_1.$ Then, a single-gate random channel drawn by local index $\ell_i\sim p_i$ is a random variable $\widehat{\mathcal U}^{(i)}
:= \|\boldsymbol\gamma_i\|_1
\,\mathrm{sign}(\gamma_i(\ell_i))\,\mathcal C_i(\ell_i)$
which is unbiased in the sense that $\mathbb E[\widehat{\mathcal U}^{(i)}]=\mathcal U^{(i)}$.

Taking the product over all local decompositions yields a circuit-level QPD
\begin{equation}
\mathcal U_{\mathrm{circ}}
= \sum_{\underline\ell\in\mathbf I}
g(\underline\ell)\,\mathcal U(\underline\ell),
\label{eq:circ-QPD-main}
\end{equation}
where $\underline\ell=(\ell_1,\dots,\ell_\nu)$ is a multi-index labelling a configuration in the Cartesian product $\mathbf I := I_1\times\cdots\times I_\nu$, the coefficients factorise as $g(\underline\ell)=\prod_{i=1}^{\nu}\gamma_i(\ell_i)$, and the associated circuit variant can be expanded as $ \mathcal U(\underline\ell):= \mathcal C_\nu(\ell_\nu)\circ\cdots\circ\mathcal C_1(\ell_1).$

The circuit-level 1-norm inherits the product structure $\|g\|_1
:= \sum_{\underline\ell\in\mathbf I} |g(\underline\ell)|
=  \prod_{i=1}^{\nu}\|\boldsymbol\gamma_i\|_1$. This induces a product distribution over the categorical distribution for every configuration
\begin{equation}
p(\underline\ell)
:= \frac{|g(\underline\ell)|}{\|g\|_1}
= \prod_{i=1}^{\nu}p_i(\ell_i),
\label{eq:p-ell-main}
\end{equation}
corresponding to independently sampling each local index $\ell_i$ from $p_i$. Product-form QPDs are the standard construction in literature \cite{van_den_berg_probabilistic_2023, koczor_probabilistic_2024}. Because the number of circuit variants in \eqref{eq:circ-QPD-main} grows exponentially with the depth $\nu$, brute force enumeration of all variants becomes rapidly infeasible, thus motivating a Monte-Carlo sampling based approach.

\subsection{Per-sample estimator, configuration samples, and repetitions per configuration}
\label{subsec:per-sample-estimator}

We now embed the circuit-level QPD into an observable-estimation task and distinguish two natural resources:

\begin{itemize}
\item the number of \emph{configurations} (channels) sampled from $p(\underline\ell)$, denoted by $K$;
\item the number of \emph{Born-rule repetitions per configuration}, denoted by $R$.
\end{itemize}

A full experiment therefore uses $N = K R$ hardware executions (assuming that all circuit variants receive $R$ repeats uniformly). 

Let $\{M_x\}$ be a POVM and write the observable as $O = \sum_x O_x M_x$
with real eigenvalues $O_x\in\mathbb R$. We assume $O$ is bounded, so that $|O_x|\le \lVert O\rVert_{\infty}$ for a finite $\lVert O\rVert_{\infty}$  (for Pauli observables, $\lVert O\rVert_{\infty}=1$). The target quantity is the ideal expectation under the original circuit,
\[
\mu
:= \Tr\bigl[O\,\mathcal U_{\mathrm{circ}}(\rho_0)\bigr].
\]
For the $k$\textsuperscript{th} configuration draw, we sample a circuit variant $\underline\ell_k \sim p(\underline\ell)$ from~\eqref{eq:p-ell-main}. Conditioned on $\underline\ell_k$, we implement $\mathcal U(\underline\ell_k)$ on $\rho_0$, and measure with $\{M_x\}$ to obtain a raw outcome $O_x$. For example, for Pauli observables $O_x \in \{ \pm 1\}$. This is then scaled by the standard QPD weight $w(\underline{\ell}_k) = \|g\|_1 \sign(g(\underline{\ell_k}))$ to obtain the \emph{single-shot per configuration} value 
$\tilde{Y}_k := w(\underline{\ell}_k)\,O_x $.

Repeating $R$ times i.i.d. we obtain the outcomes $\tilde{Y}_{k, 1}, \dots \tilde{Y}_{k, R}$, which we then average to obtain the \emph{per-configuration} mean
\begin{equation}
Y_{k}^{(R)} := \frac{1}{R}\sum_{r=1}^R \tilde Y_{k, r}
\label{eq:Yk-QPD}
\end{equation}
The natural Monte Carlo estimator for $\mu$ is then the further averaging over the configurations.
\begin{equation}
\widehat{\mu}
:=\widehat{Y}_K^{(R)} =\frac{1}{K}\sum_{k=1}^{K} Y_{k}^{(R)}.
\label{eq:YK-main}
\end{equation}

Since $|O_x|\le \ \lVert O\rVert_{\infty}$ and $|w(\underline\ell_k)|=\|g\|_1$ we have
\[
|Y_{k}^{(R)}| \le \ \lVert O\rVert_{\infty}\,\|g\|_1
\quad\Longrightarrow\quad
\Var(Y_{k}^{(R)}) \le \ \lVert O\rVert_{\infty}^2\,\|g\|_1^2.
\]
The estimator $\widehat\mu$ is unbiased for all $K$ and $R$ by linearity and the definition of the QPD weights (Appendix~\ref{app:exact-enumeration}). Moreover, if $\|\boldsymbol\gamma_i\|_1\approx \alpha$ uniformly, then $\|g\|_1\approx \alpha^\nu$, yielding the usual exponential QPD overhead \cite{cai_quantum_2023}. This can be understood as coming from a deliberately introduced `sign problem' due to the QPD re-weighting needed to ensure unbiasedness. In what follows we call the i.i.d.\ configuration sampler~\eqref{eq:p-ell-main} as the ``naïve'' sampler and illustrate how one can reduce its \emph{classical} variance component (a constant prefactor at fixed $K,R$).

\subsection{Variance decomposition and naïve design}
\label{subsec:variance-decomp}
There are two sources of randomness in the final estimator $\widehat\mu$ (or $\widehat{Y}_K^{(R)}$) in \eqref{eq:YK-main}:

\begin{itemize}
\item classical randomness over configurations $\underline\ell_k$;
\item Born-rule noise at fixed configuration $\underline\ell_k$.
\end{itemize}

The key observation underlying our approach is that, for product-form QPDs, the first source arises from an artificial and highly structured classical randomization over circuit variants. Unlike Born-rule shot noise, this configuration variance is not intrinsic and can therefore be systematically reduced---independent of the QPD construction---by classical sampling design, with stratification providing a universal and provably safe primitive.

To see this, condition the baseline estimator on $\underline\ell$, then use the law of total variance~\cite[Theorem 4.4.7]{casella_statistical_2002} to obtain
\begin{equation}
\Var(Y^{(R)})=\frac{1}{R}\,\mathbb E_{\underline\ell}\!\big[\Var(\tilde Y\mid \underline\ell)\big]
+\Var_{\underline\ell}\!\big(\mathbb E[\tilde Y\mid \underline\ell]\big),
\label{eq:var-singleconfig-QPD}
\end{equation}
and with $K$ i.i.d.\ configuration draws,
\begin{equation}
\Var(\widehat{Y}_K^{(R)})=\frac{1}{K}\Var(Y^{(R)}).
\label{eq:var-KR-QPD}
\end{equation}
The first term of~\eqref{eq:var-singleconfig-QPD} is Born-rule noise, suppressed as $1/R$; the second term is from variance between different configurations.

We will compare sampling designs at the level of total variance $\Var(Y^{(R)})$, with $R$ regarded as part of the measurement model. In particular, changing the configuration-sampling scheme can only affect the corresponding \emph{configurational} contribution $\Var_{\underline\ell}(\mu_{\underline\ell})$; the Born-rule term is fixed by the choice of $R$, circuit, and measurement. 

Naïve sampling corresponds to drawing configurations i.i.d.\ from $p(\underline\ell)$, ignoring any additional structure of the underlying distribution. The next subsection introduces stratified sampling over configurations and shows how to use additional structure to guarantee variance reduction at fixed $K$ and $R$. Henceforth, we use $Y$ to denote the per-configuration average $Y^{(R)}$ under a fixed measurement model $R$.

\subsection{Stratified sampling over QPD configurations}
\label{subsec:stratified-qpd-main}

Let $S=S(\underline\ell)$ be a deterministic function of the configuration index. We promote $S$ to a random variable via the product distribution $\underline\ell\sim p(\underline\ell)$, and define the shorthand
\[
w_s := \Pr(S=s),\,
\mu_s := \mathbb E[Y|S=s],\,
\sigma_s^2 := \Var(Y| S=s).
\]
Such that $\sum_s w_s \mu_s = \mu$. We call each such possible value of $S$ a \emph{stratum}. A \emph{stratified} estimator fixes in advance how many configuration draws are allocated to each stratum. 

We now show how this reduces the variance. Let $K_s$ be the number of configuration draws with $S=s$, such that $K_s\ge 0$ and $\sum_s K_s = K.$ Then, within each stratum we form a sample mean, and then combine these in turn to create the stratified estimator
\begin{equation}
\widehat\mu_s := \frac{1}{K_s}\sum_{t=1}^{K_s}Y_{s,t}, \quad \widehat{Y}_K^{\mathrm{strat}}
:= \sum_s w_s\,\widehat\mu_s.
\label{eq:Ystrat-QPD}
\end{equation}
By construction $\mathbb E[\widehat\mu_s]=\mu_s$ and hence $\mathbb E[\widehat{Y}_K^{\mathrm{strat}}]=\mu$ (unbiasedness proofs are in Appendix~\ref{app:exact-enumeration}).

Assuming independent sampling across strata, the variance of the stratified estimator is given by
\begin{equation}
\Var(\widehat{Y}_K^{\mathrm{strat}})
= \sum_s \frac{w_s^2\,\sigma_s^2}{K_s}, 
\label{eq:var-strat-general-QPD}
\end{equation}
So we see that the total variance depends on the number of samples we assign to each stratum. There are two particularly natural ways of budgeting: proportional allocation with $K_s \propto w_s$, and Neyman allocation with $K_s \propto w_s\sigma_s$. Neyman allocation minimises~\eqref{eq:var-strat-general-QPD} over $\{K_s\}$ but requires (pilot) estimates of the within-stratum standard deviations $\sigma_s$. In this work we focus on proportional allocation, which depends only on the stratum weights $w_s$ and already enjoys a ``never-worse'' guarantee relative to naïve sampling; Neyman allocation and pilot schemes are discussed in Appendix~\ref{app:cost-pilots}. 

Under ideal proportional allocation, we have the following theorem.

\begin{theorem}[Proportional Stratification is Never Worse than Naïve Sampling]
\label{theorem:proportional variance}
Consider a QPD estimator $\widehat{Y}_K$ with a total budget of $K$ configurations. Let $w_s$ be the probability of a configuration falling into stratum $s$, and let $\mu_s$ and $\sigma_s^2$ denote the mean and variance of the estimator within that stratum, respectively. 

Under an idealized proportional allocation $K_s = K w_s$ (ignoring integrality), the variance reduction compared to naïve sampling is exactly proportional to the \textbf{between-stratum variance}:
    \begin{align}
    \Var(\widehat{Y}_K^{\mathrm{naive}}) - \Var(\widehat{Y}_K^{\mathrm{prop}}) &= \frac{1}{K} \sum_s w_s (\mu_s - \bar{\mu})^2 \\ \nonumber
    &= \frac{1}{K} \Var_s(\mu_s) \ge 0.
	\end{align}
Consequently, proportional stratification provides a strictly lower variance whenever the stratum means $\mu_s$ are not all identical.
\end{theorem}

\begin{proof}
We use the Law of Total Variance to decompose the total variance of a single configuration sample $Y$ into ``within-stratum'' and ``between-stratum'' components:
\begin{equation}
\Var(Y) = \underbrace{\sum_s w_s \sigma_s^2}_{\text{Within}} + \underbrace{\sum_s w_s (\mu_s - \bar{\mu})^2}_{\text{Between}}.
\label{eq:total-variance-decomp}
\end{equation}
For $K$ independent naïve samples, the variance is simply $\Var(\widehat{Y}_K^{\mathrm{naive}}) = \frac{1}{K} \Var(Y)$. 

In contrast, the stratified estimator $\widehat{Y}_K^{\mathrm{prop}} = \sum_s w_s \widehat{\mu}_s$ (where $\widehat{\mu}_s$ is the sample mean of the $K_s$ shots taken in stratum $s$) has variance:
\begin{equation}
\Var(\widehat{Y}_K^{\mathrm{prop}}) = \sum_s w_s^2 \frac{\sigma_s^2}{K_s} = \sum_s w_s^2 \frac{\sigma_s^2}{K w_s} = \frac{1}{K} \sum_s w_s \sigma_s^2.
\label{eq:prop-variance-result}
\end{equation}
Subtracting \eqref{eq:prop-variance-result} from \eqref{eq:total-variance-decomp} yields the claim.
\end{proof}

\textbf{Practical Caveat.} The guarantee  in Theorem~\ref{theorem:proportional variance} is exact for ideal quotas. However, in general, $K_s = w_s K$ is not an integer, and so a rounding must be applied. In Appendix~\ref{app:allocation details} we bound the perturbation introduced by this rounding with a computable certificate, and in Appendix~\ref{app:numerical-methodology} find that this is empirically negligible in the reported numerics.

Theorem~\ref{theorem:proportional variance} can also be explicitly cast in terms of sample complexity through the following corollary.

\begin{corollary}[Sample complexity under proportional stratification]
\label{cor:sample complexity}
Fix the measurement model (i.e.\ fixed $R$) and let $\widehat{Y}_{K}^{\mathrm{naive}}$ and
$\widehat{Y}_{K}^{\mathrm{prop}}$ denote the naïve and proportional stratified estimators based on $K$
configurations. For any target standard error $\varepsilon^2>0$, define
\begin{align*}
K_{\mathrm{naive}}(\varepsilon^2)&:=\min\{K:\Var(\widehat{Y}_{K}^{\mathrm{naive}})\le \varepsilon^2\}\\
K_{\mathrm{prop}}(\varepsilon^2)&:=\min\{K:\Var(\widehat{Y}_{K}^{\mathrm{prop}})\le \varepsilon^2\}.
\end{align*}

Under ideal proportional quotas, 
\[
\frac{K_{\mathrm{prop}}(v)}{K_{\mathrm{naive}}(v)}=
\frac{\Var(Y^{\mathrm{prop}})}{\Var(Y^{\mathrm{naive}})}
=\rho\in[0,1].
\]
Equivalently, to achieve the same target error in a given estimation task, proportional stratification reduces the required configuration budget (and hence total executions $N=KR$) by a constant factor $0\le \rho\le 1$. 
\end{corollary}

\textbf{Interpretation (explained configuration variance).}
The statistic $S(\underline\ell)$ provides a compact description of the configuration
$\underline\ell$ and therefore induces stratum-conditional means
$\mu_s:=\mathbb E[Y\mid S=s]$.
The law of total variance decomposes the single-configuration variance as
\[
\Var(Y)=\mathbb E[\Var(Y\mid S)]\;+\;\Var(\mathbb E[Y\mid S]),
\]
so the between-stratum term $\Var(\mathbb E[Y\mid S])$ is precisely the portion of variability in $Y$
that is \emph{explained} by knowing $S$.
Under ideal proportional allocation, stratification removes exactly this explained component, leaving only
the within-stratum contribution $\mathbb E[\Var(Y\mid S)]$.
Equivalently, defining the explained fraction
\[
R^2_{\mathrm{eff}}(S):=\frac{\Var(\mathbb E[Y\mid S])}{\Var(Y)}\in[0,1],
\]
the ideal variance ratio satisfies $\rho(S)=1-R^2_{\mathrm{eff}}(S)$: the improvement $1-\rho$
is exactly the fraction of single-configuration variance captured by $S$.
In the oracle limit (no shot noise), $Y$ is deterministic given $\underline\ell$ and $R^2_{\mathrm{eff}}(S)$
reduces to the usual fraction of configuration-to-configuration mean variation explained by $S$; in this case,
achieving the same precision requires only $K_{\mathrm{prop}}=\rho(S)\,K_{\mathrm{naive}}$ configurations.
With finite shots, Born noise increases $\Var(Y)$ without increasing $\Var(\mathbb E[Y\mid S])$, reducing
$R^2_{\mathrm{eff}}(S)$ and hence the attainable savings. Appendix~\ref{app:sufficiency} formalises these identities
and connects them to permutation symmetry for counts-vector strata.

\paragraph{Remark (QPD overhead unchanged).}
Stratification changes only how configurations are selected and aggregated; it does not modify the
per-shot weight magnitude. In particular, for bounded observables $|O_x|\le \ \lVert O\rVert_{\infty}$ and our
choice $|w(\underline\ell)|=\|g\|_1$, every realised per-configuration average satisfies
$|Y^{(R)}|\le \ \lVert O\rVert_{\infty}\|g\|_1$, irrespective of whether $\underline\ell$ was drawn i.i.d.\ or conditionally
within a stratum. Consequently, we have the upper bound
\[
\Var(\widehat{Y}_K^{\mathrm{strat}})\le \frac{\ \lVert O\rVert_{\infty}^2\|g\|_1^2}{K},
\]
so the fundamental $1/K$ and $\|g\|_1^2$ scaling (and its typical exponential growth in $\nu$) is unchanged; the effect of stratification is a reduction of the constant pre-factor, quantified by the variance ratio $\rho$ (it only removes the additive between-stratum contribution $\Var_s(\mu_s)$).

\subsection{Counts-vector stratification for QPDs}
\label{subsec:counts-main}

Na\"{\i}ve sampling treats the configuration index $\underline\ell$ as an atomic label. However, in many QPD ensembles different configurations contribute similarly to the observable, particularly when there are exact or approximate permutation symmetries across gate locations. We exploit this by stratifying configurations using a permutation-invariant \emph{counts vector}, which is a concrete instantiation of the stratification statistic $S(\underline\ell)$ introduced above. We emphasise that the counts-vector statistic is not unique: there is a natural family of coarser statistics obtained by coarsening $\mathbf M$ (see Appendix~\ref{app:coarsening-main}).

\paragraph{Padding to a common width.}
For notational convenience we assume a common local width $d$. This can be enforced without loss of generality by padding each local decomposition with dummy primitives of zero coefficient: that is, choose $d=\max_i d_i$ and extend all local index sets $I_i$ to $\{1,\dots,d\}$ by setting $\gamma_i(\ell)=0$ for $\ell>d_i$. This leaves the induced product law $p(\underline\ell)$ and circuit 1-norm $\|g\|_1$ unchanged.

\paragraph{Counts-vector strata.}
With this notational simplification the counts vector associated with a multi-index $\underline\ell=(\ell_1,\dots,\ell_\nu)$ records how often each local index appears:
\begin{equation}
\mathbf M = (M_1,\dots,M_d),\quad
M_k := \#\{i : \ell_i = k\},
\label{eq:counts-vector}
\end{equation}
with $\sum_{k=1}^d M_k = \nu$. Configurations that differ only by permuting gate positions share the same
$\mathbf M$, so taking $S\equiv \mathbf M$ partitions the configuration space into permutation-invariant strata.
In the mixed-width case, the padded categories simply have zero probability.

For fixed $\nu$ and $d$, the number of distinct counts vectors is
\begin{equation}
\#\{\mathbf m : \sum_k m_k=\nu\}
= \binom{\nu+d-1}{d-1}
= O(\nu^{d-1}),
\label{eq:num-counts-states}
\end{equation}
which is polynomial in $\nu$ for fixed local width $d$. The special case $d=2$ recovers the binomial-type
stratifications used for homogeneous QPDs (see \cite{chen_faster_2025} and Appendix~\ref{app:related-work}). For comparison, there are $O(d^\nu)$ possible circuit configurations in the full product distribution $p(\underline\ell)$. 

\paragraph{DP preprocessing and conditional sampling.}
Operationally, stratified sampling with $S\equiv \mathbf M$ requires (i) full knowledge of the exact stratum weights
$w_{\mathbf m}:=\Pr(\mathbf M=\mathbf m)$ and (ii) an efficient method to sample configurations
$\underline\ell$ \emph{conditionally} on $\mathbf M=\mathbf m$. In the product-form QPD setting where we have a simple product of categorical distributions
$p(\underline\ell)=\prod_{i=1}^\nu p_i(\ell_i)$, the induced counts-vector random variable $\mathbf M(\underline\ell)$
follows a Poisson---multinomial distribution (PMD) determined by the local categories $\{p_i\}$.
We can exploit this PMD structure to solve both tasks via a single dynamic programme (DP) that runs as a classical pre-computation, such that quantum run-time is unaffected. This algorithm can be seen as a generalisation of the standard DP for calculating the Poisson---Binomial distribution \cite{hong_computing_2013, chen_statistical_1997}. We also note that the PMD is a well-studied distribution in classical statistics \cite{lin_poisson_2022}.

Concretely, a forward dynamic programme computes all $w_{\mathbf m}$ with $O(d\,\nu^{d})$ arithmetic operations and caches
intermediate layers in $O(\nu^{d})$ memory; a backward pass through the cached table then generates
$\underline\ell\sim p(\underline\ell\mid \mathbf M=\mathbf m)$ in $O(\nu d)$ time per configuration, matching the runtime
cost of na\"{\i}ve sampling up to a constant factor. Full recursions and pseudocode for performing the two tasks, namely
\textsc{CountsForwardDP} and \textsc{CountsConditionalSample}, are given in Appendix~\ref{app:counts-dp}.

\paragraph{Complexity caveat and routes to scalability.}
The DP state space grows as $O(\nu^{d-1})$ and the preprocessing cost as $O(d\,\nu^d)$, so while the method is efficient for the small fixed-width regime targeted in our numerics in Section~\ref{sec:numerical results} when $d$ is small (e.g. $d=3$ for PAI and $d=4$ for single-qubit Pauli PEC, and $\nu$ up to around $10^2$), it can become impractical when the width is large (e.g.\ multi-qubit noise channels or broader local decompositions like in \cite{koczor_sparse_2024}) or when $\nu$ is extremely large. This is not a conceptual failure of stratification, but an implementation constraint of the \emph{exact} PMD/conditional-sampling backbone. In such regimes one can retain the same design viewpoint while introducing problem-dependent refinements: globally coarsening the statistic (reducing effective width $d$), using typed/blockwise counts to keep widths small within blocks, or using pruned/approximate pre-processing (potentially with controlled bias) to exploit concentration of the counts-vector distribution. We sketch these refinements in the outlook, where the goal becomes an explicit trade-off between pre-processing cost and variance (and possibly bias).

\paragraph{Generality and problem-agnosticism.}
The counts-vector construction and the associated DP are purely distributional: they depend only on the
product-form configuration law $p(\underline\ell)=\prod_{i=1}^\nu p_i(\ell_i)$ induced by the local QPD
coefficients, and do not require additional assumptions about the quantum primitives beyond the ability to
index them. In particular, the DP never inspects (or utilises) the underlying channels $\{\mathcal C_i(\ell)\}$, the surrounding circuit that they're embedded in, or the measured observable; it uses only the categorical probabilities
$\{p_i(\cdot)\}$. Consequently, the same pre-processing and conditional sampler apply unchanged to inhomogeneous product-form QPDs where the local distributions vary with $i$, such as when different locations correspond to different physical primitives (e.g.\ dephasing at one location and bit-flip at another, or different PAI axes), so long as each location is equipped with a local index set and associated probabilities. As such, we emphasise that counts-vector stratification (or statistics coarsened from it) is a universal, quantum agnostic sampling-design layer for product-form QPDs: it is a drop-in replacement for na\"{\i}ve configuration sampling. Any quantum- or observable-specific structure enters only through the conditional means $\mu_{\underline\ell}$ and hence through how informative $\mathbf M$ is for a given problem instance.

\subsection{Practical stratified sampling recipe}
\label{subsec:practical-sampling}

For a product-form QPD and a chosen stratification statistic $S(\underline\ell)$ (here the counts vector $S\equiv\mathbf M$), stratified sampling replaces i.i.d.\ draws of configurations $\underline\ell\sim p(\underline\ell)$ by the following end-to-end procedure:

\begin{enumerate}
  \item \textbf{Precomputation (once per QPD-problem):}
  run \textsc{CountsForwardDP} to obtain the stratum weights $w_{\mathbf m}=\Pr(\mathbf M=\mathbf m)$ and cache the prefix tables $\{W^{(i)}_{\mathbf m}\}$ required for backward conditional sampling. 

  \item \textbf{Integer allocation across strata:}
  given a total configuration budget $K$, choose nonnegative integers $K_{\mathbf m}$ with $\sum_{\mathbf m}K_{\mathbf m}=K$.
  To preserve exact unbiasedness at finite $K$, we use residual-aware Hamilton apportionment: any zero-allocation strata are aggregated into a residual bucket $\mathcal{D}$ with total weight $w_*$ and count $K_*$ (see Appendix~\ref{app:allocation details} for details).

  \item \textbf{Stratified configuration generation and circuit execution:}
  \begin{enumerate}
    \item For each stratum $\mathbf m$ with $K_{\mathbf m}>0$, repeat $K_{\mathbf m}$ times:
    sample $\underline\ell\sim p(\underline\ell\mid \mathbf M=\mathbf m)$ via the backward DP, execute $\mathcal U(\underline\ell)$, and compress $R$ measurement repetitions into a single scalar outcome $Y_{\mathbf m,t}$ (Sec.~\ref{subsec:per-sample-estimator}).
    \item If there are residuals ($K_*>0$) repeat $K_*$ times:
    sample a stratum label $\mathbf m\sim q(\cdot)$ on the residual set $\mathcal D$, then sample $\underline\ell\sim p(\underline\ell\mid \mathbf M=\mathbf m)$ and execute as above to obtain $Y_{*,t}$.
  \end{enumerate}

  \item \textbf{Estimator construction:}
  compute within-stratum means $\widehat\mu_{\mathbf m}=\frac{1}{K_{\mathbf m}}\sum_{t=1}^{K_{\mathbf m}}Y_{\mathbf m,t}$ (and $\widehat\mu_*$ analogously if $K_*>0$), and return
  \[
  \widehat{Y}_{K,R}^{\text{impl}}
  = \sum_{\mathbf m} w_{\mathbf m}\,\widehat\mu_{\mathbf m}
  + w_*\,\widehat\mu_* .
  \]
\end{enumerate}
\begin{algorithm}[t]
\caption{\textsc{StratifiedQPDMC} (counts-vector strata; residual-aware proportional allocation)}
\label{alg:stratified-qpdmc}
\KwIn{Local probabilities $\{p_i(k)\}_{i\le\nu,k\le d}$; total configuration budget $K$; repetitions per configuration $R$.}
\KwOut{Estimator $\widehat{Y}_{K,R}^{\mathrm{impl}}$.}

$(\{w_{\mathbf m}\},\ \mathcal T)\leftarrow \textsc{CountsForwardDP}(\{p_i\})$\tcp*[r]{Appendix~\ref{app:counts-dp}}
$(K_{\mathbf m},K_*,\mathcal D,w_*,q)\leftarrow \textsc{ResidualHamiltonAllocate}(\{w_{\mathbf m}\},K)$\tcp*[r]{Appendix~\ref{app:allocation details}}

\ForEach{counts vector $\mathbf m$ with $K_{\mathbf m}>0$}{
  \For{$t\leftarrow 1$ \KwTo $K_{\mathbf m}$}{
    $\underline\ell \leftarrow \textsc{CountsConditionalSample}(\mathbf m,\{p_i\},\mathcal T)$\tcp*[r]{Appendix~\ref{app:counts-dp}}
    Run $\mathcal U(\underline\ell)$ and compress $R$ repetitions into $Y_{\mathbf m,t}$\;
  }
  $\widehat\mu_{\mathbf m}\leftarrow \frac{1}{K_{\mathbf m}}\sum_{t=1}^{K_{\mathbf m}} Y_{\mathbf m,t}$\;
}

\If{$K_*>0$}{
  \For{$t\leftarrow 1$ \KwTo $K_*$}{
    Sample $\mathbf m \sim q(\cdot)$ on $\mathcal D$\;
    $\underline\ell \leftarrow \textsc{CountsConditionalSample}(\mathbf m,\{p_i\}, \mathcal T)$\;
    Run $\mathcal U(\underline\ell)$ and form $Y_{*,t}$\;
  }
  $\widehat\mu_*\leftarrow \frac{1}{K_*}\sum_{t=1}^{K_*} Y_{*,t}$\;
}
\Else{$\widehat\mu_*\leftarrow 0$\;}

\Return $\widehat{Y}_{K,R}^{\mathrm{impl}}=\sum_{\mathbf m: K_{\mathbf{m}} > 0} w_{\mathbf m}\widehat\mu_{\mathbf m} + w_*\,\widehat\mu_*$\;
\end{algorithm}

Algorithm~\ref{alg:stratified-qpdmc} gives pseudocode for the full routine.
Under ideal proportional quotas $K_{\mathbf m}=Kw_{\mathbf m}$, the corresponding proportional estimator satisfies the never-worse guarantee in Theorem~\ref{theorem:proportional variance}. The residual-aware implementation preserves exact unbiasedness at finite $K$ and, in our numerics, induces variance perturbations smaller than the reported bootstrap uncertainties. As such, our numerical simulations substantiate the guarantee of Theorem~\ref{theorem:proportional variance} on two specific problems and choice of statistic.

\section{Numerical Results \& Discussion}
\label{sec:numerical results}

We benchmark counts-vector stratification for product-form QPDs using first-order Trotter time evolution of the 1D transverse-field Ising model (TFIM) and two representative QPD schemes: probabilistic angle interpolation (PAI) and probabilistic error cancellation (PEC). Throughout we work at the configuration level: each sampled QPD variant (together with a specified measurement protocol) produces a single scalar outcome $Y$, possibly by compressing $R$ repeated shots via the Born rule. 

For a fixed stratification statistic $S$ (here the counts vector $\mathbf M$), we compare naïve and proportional designs through the design-variance ratio (see Corollary~\ref{cor:sample complexity})
\begin{equation}
\rho(R)
:= \frac{\Var(\widehat{Y}_{K,R}^{\mathrm{impl}})}{\Var(\widehat{Y}_{K,R}^{\mathrm{naive}})}
\in[0,1],
\label{eq:variance ratio}
\end{equation}
where $R$ is the number of Born-rule repetitions per configuration and $\widehat{Y}_{K,R}^{\mathrm{impl}}$ denotes our residual-aware proportional estimator (Algorithm~\ref{alg:stratified-qpdmc}). In the ideal proportional model $\rho(R)$ is independent of $K$; empirically, the finite-$K$ apportionment effects are negligible in our numerics with $K=8192$ (Appendix~\ref{app:numerical-methodology} and Appendix~\ref{app:var-rounding-residual}). Operationally, for fixed $R$, a ratio $\rho(R)< 1$ corresponds to a constant-factor reduction in the required configuration budget $K$ (and hence total executions $N=KR$ assuming each configuration gets $R$ shots) to reach a target precision. We emphasise that stratification \emph{cannot} alter the asymptotic performance.

\subsection{Measurement models and interpretation}
\label{subsec:numerics-models}

We utilise three measurement models:
(i) a \emph{single-shot} model ($R=1$), where each configuration is measured once and $Y$ contains both configuration randomness and Born-rule noise;
(ii) an \emph{oracle} model, where each realised configuration $Y = \underline\ell$ is assigned its exact conditional mean $\mu_{\underline\ell}$ (equivalently the $R\to\infty$ limit);
and (iii) an intermediate choice $R=64$, which suppresses the Born-rule contribution by $\sim 1/R$. This intermediate choice shows how $\rho(R)$ interpolates between the oracle and single-shot limits.

\subsection{TFIM Trotter benchmark}
\label{subsec:tfim-trotter}

We simulate first-order (Lie--Trotter) time evolution for the TFIM Hamiltonian on a ring of $n$ spins,
\[
H=h\sum_{j=1}^n X_j + J\sum_{j=1}^{n} Z_j Z_{j+1},
\quad Z_{n+1}\equiv Z_1,
\]
with $h=0.6$, $J=0.7$, total time $t=1.0$, initial state $\rho_0=|0\rangle\!\langle 0|^{\otimes n}$, and observable $\langle X_1\rangle$.
We fix $n=6$ and vary the Trotter depth $L$. For each $L$ and each design (naïve versus proportional-stratified) we draw $K=8192$ configurations, estimate design variances in the oracle and single-shot models, and form empirical ratios with bootstrap confidence bands (Appendix~\ref{app:numerical-methodology}). In Appendix~\ref{subsec:abs-variances} we also report the corresponding normalised absolute variances $\Var(\widehat Y_{K,R})/\|g\|_1^2$, whose ratios recover the $\rho(R)$ curves shown here.

\subsection{PAI and PEC: variance ratios versus depth}
\label{subsec:pai-pec}

\paragraph{Probabilistic angle interpolation (PAI).}
In PAI \cite{koczor_probabilistic_2024}, each rotation is replaced by a three-term local QPD over nearby hardware-available angles. Over $\nu$ rotations, configurations $\underline\ell\in\{1,2,3\}^\nu$ induce counts-vector strata
\[
\mathbf M=(M_1,M_2,M_3),
\quad \sum_{i=1}^3 M_i =\nu,
\]
where each $M_i$ counts up the occurrences of primitive index $i$ in a given configuration, ignoring all order.
Fig.~\ref{fig:trotter PAI} shows the variance ratios versus depth. In the single-shot model, stratification yields a modest but consistent improvement ($\widehat\rho^{\mathrm{shot}}\approx 0.9$). In the oracle model, the same statistic yields substantially larger reductions ($\widehat\rho^{\mathrm{oracle}}\approx 0.3$--$0.6$), indicating that counts capture a nontrivial fraction of configuration-level structure. The gradual degradation of $\widehat\rho^{\mathrm{oracle}}$ with depth is consistent with increasing non-commutation and position-dependence, which makes configurations within a fixed counts stratum less homogeneous.

\begin{figure}[tb]
    \centering
    \includegraphics[width=\linewidth]{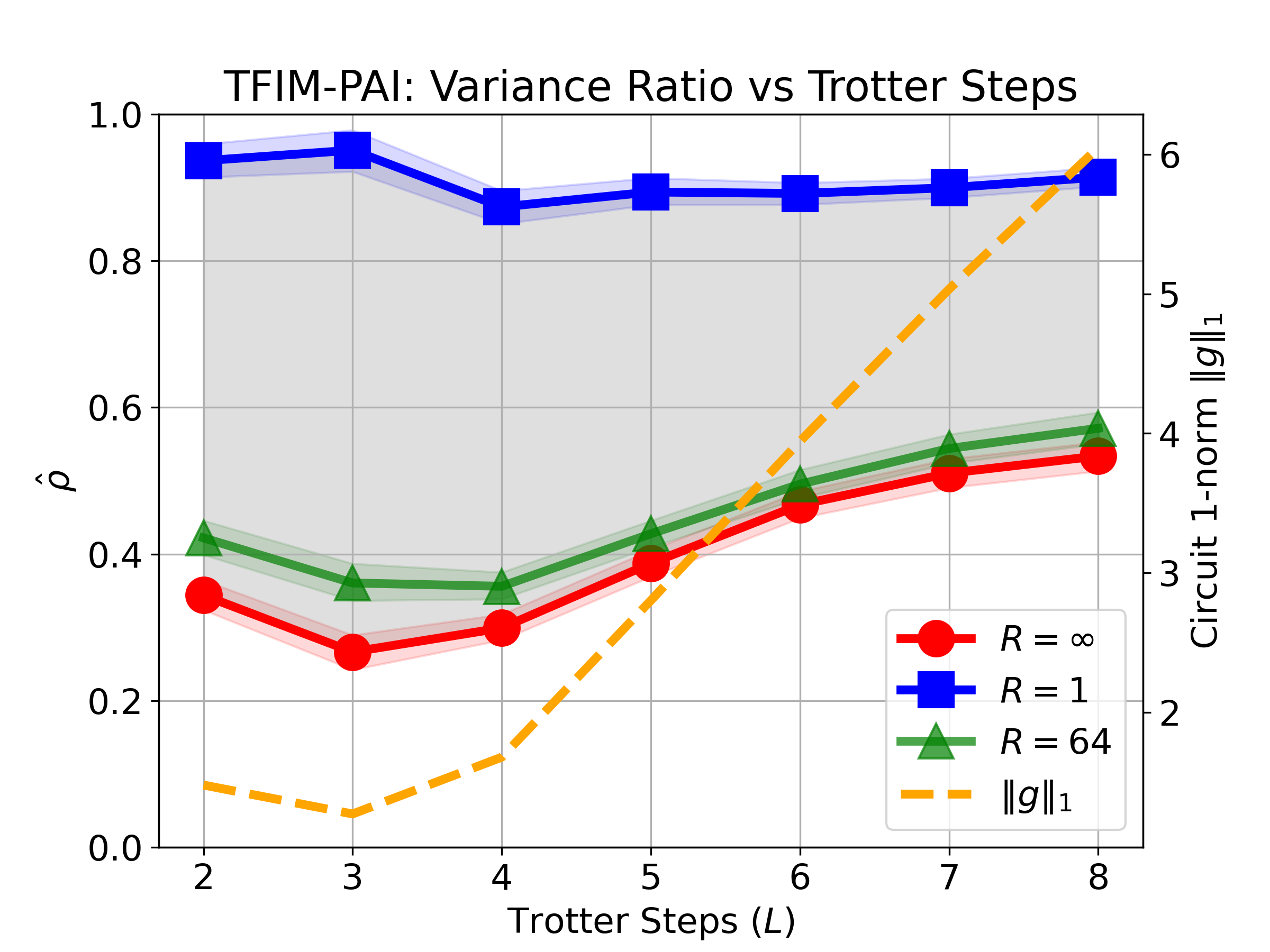}
    \caption{PAI on TFIM Trotter circuits ($n=6$, $h=0.6$, $J=0.7$, $t=1.0$) estimating $\langle X_1\rangle$. For each depth $L$ we draw $K=8192$ configurations and report empirical variance ratios under naïve and counts-vector proportional designs in the oracle (red) and single-shot (blue) models; the green curve shows $R=64$. Bands denote bootstrap confidence intervals. The dashed orange curve shows the circuit QPD 1-norm $\|g\|_1$ on a secondary axis. The circuit depths are given by $\nu = 12L$.}
    \label{fig:trotter PAI}
\end{figure}

\paragraph{Probabilistic error cancellation (PEC).}
In PEC \cite{temme_error_2017}, each ideal gate is accompanied by an inverse noise channel expressed as a local QPD over implementable noisy channels. We consider gate-independent single-qubit depolarising noise of strength $p=0.01$ after each unitary (applied on both qubits for two-qubit gates), giving a fixed-width local QPD over Pauli channels $\{I,X,Y,Z\}$. Configurations $\underline\ell\in\{I,X,Y,Z\}^\nu$ induce strata 
\[
\mathbf M=(M_I,M_X,M_Y,M_Z),
\quad \sum_{P\in\{I,X,Y,Z\}} M_P=\nu.
\]
Details of the PEC setup are given in Appendix~\ref{app:qpd-instances}. Fig.~\ref{fig:trotter PEC} shows the empirical variance ratios. As for PAI, single-shot improvements are modest (typically $\widehat\rho^{\mathrm{shot}}\approx 0.9$), whereas the oracle ratios are substantially smaller (typically $\widehat\rho^{\mathrm{oracle}}\approx 0.2$--$0.4$), demonstrating that counts-vector stratification can remove a large fraction of the classical configuration variance even in the higher-overhead PEC setting.

\paragraph{Why the counts-vector statistic works.}
For proportional stratification with statistic $S$, the ideal theory gives
\(
\Var(Y^{\mathrm{naive}})-\Var(Y^{\mathrm{prop}})=\Var_s(\mu_s)\)
so the gain is governed by how much variation in configuration-level conditional means is explained by conditioning on $S$.
In the oracle model this simplifies to
\(
\rho^{\mathrm{oracle}}(S)=1-R^2(S),
\)
where $R^2(S)$ is the explained fraction of configuration-to-configuration mean variation. In the TFIM--PAI and TFIM--PEC benchmarks we observe $1-\rho^{\mathrm{oracle}}$ ranging from $0.26$ up to $0.53$ over depths $L=2,\dots,8$; see Appendix~\ref{app:sufficiency} for further discussion on this `explained variance' perspective. Put simply, counts-vector stratification helps most when configurations that share the same counts vector also have similar conditional means $\mu_{\underline{\ell}}$, such that most of the configurational variance \(\Var_{\underline\ell}(\mu_{\underline\ell})\) is ``explained'' by $\mathbf M$. 

Separately, among statistics that are permutation-invariant in the sense of depending only on multiplicities (and not on gate positions), the full counts vector is variance-minimal within the monotone coarsening hierarchy; see Appendix~\ref{app:coarsening-main} (and the concrete counts-versus-parity comparison therein). We emphasise that stratification on coarser statistics (such as sign-parity) can be easily recovered from the counts-vector.  

For $d\le 4$ and $\nu$ up to $~10^2$, the DP Algorithm~\ref{alg:counts-forward} is practical on a personal workstation; for much larger problems coarsening or pruning will be required. In our implementation we elected to preserve unbiasedness and generality and were consequently limited to modest-depth simulations. We note that introducing a controlled bias can both yield further variance reductions (in the sense of mean-squared error), while simultaneously also improving the pre-compute complexity of the DP and conditional sampling algorithms by allowing for approximations instead. Generalising the results presented here into the bias-variance context is a promising direction for future work.

\section{Conclusion and Outlook}
\label{sec:outlook}
Product-form QPD protocols introduce a purely classical source of uncertainty on top of intrinsic Born-rule shot noise: \emph{configuration variance} which arises from randomising over circuit variants. Stratification helps reduce this additional variance. While many QPD protocols were originally motivated by NISQ-era constraints (e.g., error mitigation via circuit randomisation), circuit-level randomisation is now a broadly recurring motif across quantum algorithms and benchmarking: whenever an expectation value is ultimately obtained by sampling, one often benefits from introducing and then properly accounting for ensemble randomness.

This remains true in early fault-tolerant settings, where shot budgets may be even tighter due to slow logical clock rates: even modest constant-factor reductions in the required number of configurations or shots can materially improve throughput especially on first-generation fault-tolerant machines.

We cast these protocols as a two-level Monte Carlo estimator and used the law of total variance to isolate the configuration contribution. This viewpoint yields a general sampling-design principle: for any statistic $S(\Lambda)$ on the configuration space, proportional stratification is unbiased and (under ideal proportional quotas) almost always decreases variance relative to naïve sampling at fixed configuration budget.

\begin{figure}[tb]
	\centering
	\includegraphics[width=\linewidth]{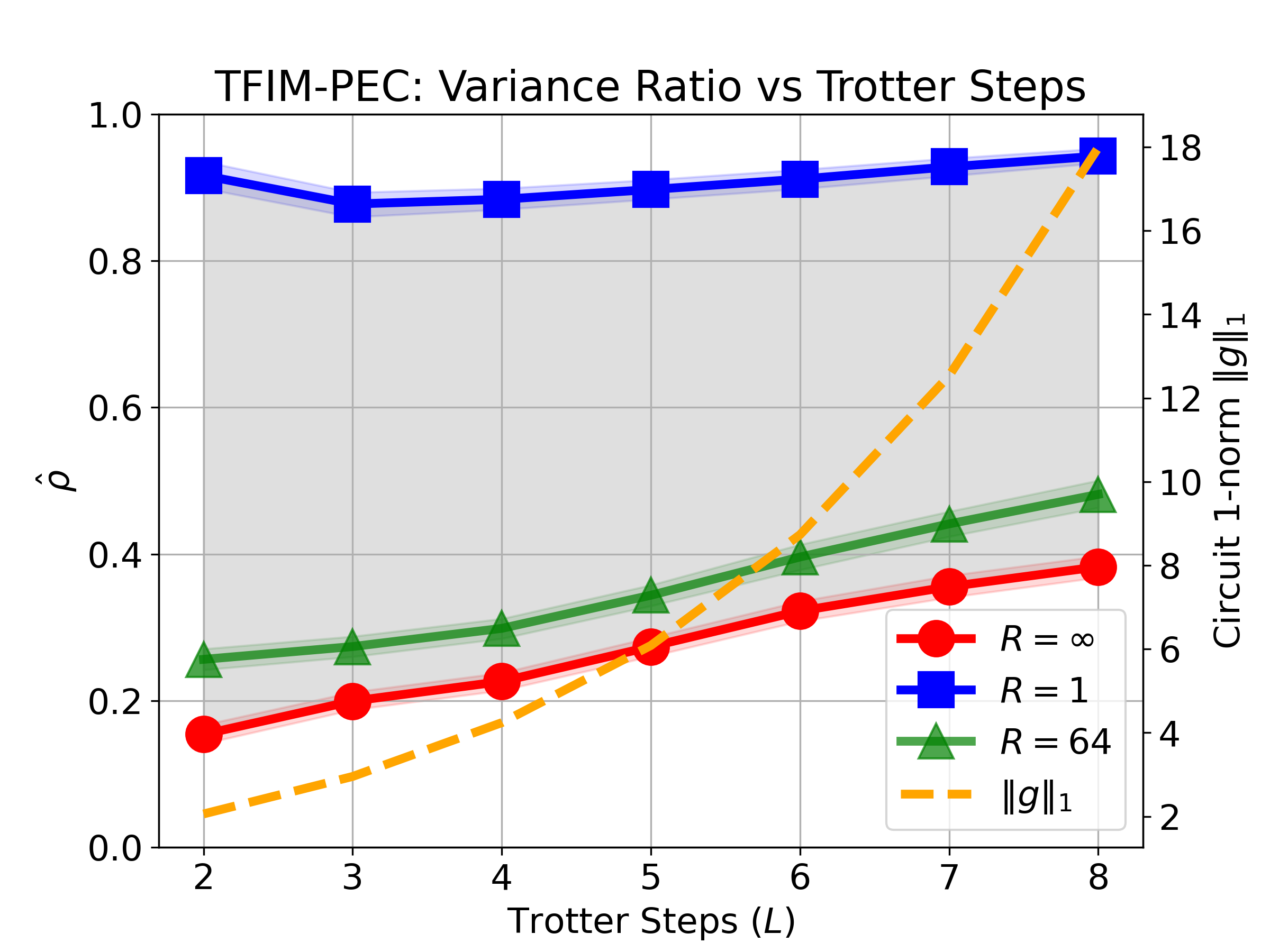}
	\caption{PEC on the same TFIM Trotter circuits as in Fig.~\ref{fig:trotter PAI}, with gate-independent single-qubit depolarising noise ($p=0.01$) after each unitary. We draw $K=8192$ configurations per $L$ and report empirical variance ratios with bootstrap confidence bands. The dashed orange curve shows the circuit QPD 1-norm $\|g\|_1$ on a secondary axis. The circuit depths are given by $\nu = 18L$.}
	\label{fig:trotter PEC}
\end{figure}

We then instantiated this principle with a universal, permutation-invariant statistic---the \emph{counts vector}---and made the resulting stratified design operational. For fixed local width $d$, a single dynamic programme over the induced Poisson-multinomial distribution computes all stratum weights and enables exact conditional sampling, with a one-off classical cost $O(d\,\nu^d)$ and per-configuration sampling cost $O(\nu d)$. After preprocessing, the stratified routine is a drop-in replacement for naïve configuration sampling. In cases where $\nu$ or $d$ are large, the same variance reduction guarantee (Theorem~\ref{theorem:proportional variance}) holds for coarser, but cheaper statistics like sign-parity, which has a classical cost $O(\nu^2)$ independent of width $d$ (see Appendix~\ref{app:coarsening-main}).

On first-order Trotter benchmarks with probabilistic angle interpolation (PAI) and probabilistic error cancellation (PEC), counts-vector stratification consistently reduces configuration-level variance: improvements are modest but robust in Born-dominated single-shot regimes (about $10\%$), and substantially larger in oracle-like regimes where configuration variance dominates (up to $\sim 60$--$80\%$ reduction in our case studies). As expected, these gains do not alter the exponential overhead set by $\|g\|_1^2$, but they reduce constant prefactors and therefore expand the practically useful range of existing QPD methods.

Looking forward, the same variance-accounting framework suggests clear upgrade paths: (i) statistics that deliberately break permutation symmetry to explain more configuration-to-configuration variation, (ii) pilot schemes approximate variance-optimal (Neyman) allocation across strata, or on-line adaptive schemes to achieve further cost optimisations, (iii) scalable approximate preprocessing and controlled truncation for deeper circuits, and (iv) extensions to hierarchical settings with simultaneous randomisations, such as classical shadows or circuit cutting on top of QPD schemes. In this last case, a hierarchical stratification (derived from recursively applying the law of total variance) will be a natural way of reducing total resource costs.  

To summarise, our contribution here is to separate ensemble randomness from Born-rule noise and to show that sampling design---independent of QPD construction---offers provable, composable constant-factor savings for product-form QPD protocols.

\textbf{Data availability}: The simulation code used is available on GitHub at \url{https://github.com/joshua-dai/stratifiedQPD}.

\begin{acknowledgments}
    We thank Simon Benjamin, Po-Wei Huang, Minjun Jeon, and Chusei Kiumi for helpful discussions. We additionally thank Ben Criger and the Cambridge Quantinuum team for stimulating initial conversations and sketching out an adaptive sampling scheme. JWD is supported by a Clarendon Fund Scholarship, University of Oxford. BK thanks UKRI for the Future Leaders Fellowship Theory to Enable Practical Quantum Advantage (MR/Y015843/1). BK also acknowledges funding from the EPSRC project Robust and Reliable Quantum Computing (RoaRQ, EP/W032635/1). This research was funded in part by UKRI (MR/Y015843/1). For the purpose of Open Access, the author has applied a CC BY public copyright licence to any Author Accepted Manuscript version arising from this submission.
\end{acknowledgments}

\bibliographystyle{apsrev4-2}
\bibliography{references} 

\appendix

\section{Related Work and Connections to Variance-Reduction Techniques}
\label{app:related-work}

There is extensive work on reducing the overhead of QPD-based protocols by improving the
\emph{decomposition itself}, typically with the goal of reducing the circuit 1-norm $\|g\|_1$ by exploiting
structure in the circuit, noise model, or observable
\cite{jiang_physical_2021, guo_quantum_2022, piveteau_quasiprobability_2022, scheiber_reduced_2025, tran_locality_2023, eddins_lightcone_2024, koczor_sparse_2024}.
Examples include physically informed constructions
\cite{jiang_physical_2021, guo_quantum_2022, piveteau_quasiprobability_2022}, specialised workflows in
Clifford-dominated regimes \cite{scheiber_reduced_2025}, and locality or light-cone arguments for pruning mitigation
on local observables \cite{tran_locality_2023, eddins_lightcone_2024}.
These works primarily address the \emph{ensemble design} problem: constructing a low-overhead QPD.

The present paper addresses the complementary \emph{sampling design} problem for a fixed
product-form QPD. Given such a decomposition, one must still decide how to sample configurations from its typically
enormous configuration space and how to aggregate outcomes into an estimator. In classical Monte Carlo, these
choices are studied as variance-reduction techniques (VRTs)
\cite{kroese_handbook_2011, cochran_sampling_1977, lohr_sampling_2021}, including conditioning/Rao--Blackwellization,
stratification, control variates, and importance sampling
\cite{liu_monte_2004, robert_rao-blackwellization_2021}.
While standard in stochastic simulation, these tools have been less systematically developed for QPD-based quantum
protocols as an explicit \emph{outer-layer} design problem over circuit configurations (distinct from inner
Born-rule sampling).

A notable recent example is CV4Quantum \cite{shyamsundar_cv4quantum_2025}, which applies control variates to reduce the
variance of PEC estimators under realistic conditions. CV4Quantum is a post-processing VRT: it starts from a fixed
dataset of circuit outcomes and constructs a lower-variance estimator without modifying the
configuration-sampling plan. By contrast, our approach modifies the configuration-sampling design itself (via
stratified sampling with exact conditional draws) while preserving unbiasedness. The two approaches are compatible
in principle. 

Alternatively, Ref.~\cite{aharonov_reliable_2025} also explicitly decomposes the variance of QP-based estimators into circuit-to-circuit and shot-to-shot contributions via the law of total variance, and observes that circuit-sampling distributions may be improved, including by splitting into complementary sub-distributions. We further develop stratification as a general sampling-design primitive for QPDs.

The closest precedent to our work is Ref.~\cite{chen_faster_2025}, which develops a PEC procedure for a gate-independent noise setting by grouping configurations into ``sectors'' according to the number of non-identity insertions and allocating sampling effort across sectors proportional to their weights. This sector statistic corresponds to a coarsening of the multinomial counts we present.

Our framework generalises this sector picture by treating stratification as a first-class sampling-design primitive
for arbitrary product-form QPDs. In particular, our contributions include:
\begin{enumerate}
\item \textbf{Counts-vector stratification for inhomogeneous product QPDs.}
For general product-form QPDs with position-dependent local supports and weights, we identify the full multinomial
counts vector as a canonical permutation-invariant statistic. For fixed local width $d$, this yields
$O(\nu^{d-1})$ strata. We give a dynamic programme that computes all stratum masses exactly and supports
\emph{exact} conditional sampling (Appendix~\ref{app:counts-dp}), yielding an unbiased drop-in replacement for naïve
i.i.d.\ configuration sampling after one-off preprocessing.

\item \textbf{A general ``never-worse'' guarantee and coarsening hierarchy.}
We show that under ideal proportional quotas, stratification on \emph{any} statistic is unbiased and never has
larger variance than naïve sampling at fixed configuration budget
(Theorem~\ref{theorem:proportional variance}). Moreover, coarsening a statistic yields a monotone variance hierarchy:
conditioning on finer information cannot increase the proportional-stratified variance
(Appendix~\ref{app:coarsening-main}). Binomial sectoring \cite{chen_faster_2025} arises as a special case.

\item \textbf{Connections to classical allocation theory.}
We connect stratified variance decompositions to proportional and Neyman allocation
\cite{cochran_sampling_1977}, and discuss pilot-based approximations to variance-optimal allocation in
Appendix~\ref{app:cost-pilots}.
\end{enumerate}

We also address a practical finite-$K$ issue: ideal proportional quotas are non-integer. We therefore use integer
apportionment together with a residual-stratum correction that preserves exact unbiasedness. Appendix~\ref{app:allocation details} gives details, and Appendix~\ref{subsec:certificate-exact} quantifies the resulting (typically small)
variance perturbation.

\textbf{Addendum.}
During the preparation of this manuscript, Myers \emph{et al.}~\cite{myers_simulating_2025} introduced a stratified
sampling approach for noise simulation that groups configurations into ``fault-count'' sectors and allocates
samples across sectors according to their weights, similar to the approach of \cite{chen_faster_2025} but for the classical simulation of quantum noise rather than PEC (we note that the QPD mathematics are common to both). Specifically, the method in \cite{myers_simulating_2025} exploits the weak-noise concentration of the configuration distribution on patterns with few fault insertions. In our framework, this fault-count statistic can be seen as a binomial coarsening of the full counts vector and fits naturally into the variance hierarchy discussed in see Appendix~\ref{app:coarsening-main}.

Conversely, our results provide a general-purpose stratification backbone for arbitrary product-form QPD ensembles,
with a universal never-worse guarantee and an exact Poisson--Multinomial preprocessing scheme that supports
conditional sampling. Appendix~D of Ref.~\cite{myers_simulating_2025} notes that a dynamic programme can be used to
compute Poisson--Binomial weights, corresponding to a one-dimensional special case of the Poisson--Multinomial
construction employed here. Our development was carried out independently.

\section{Exact moments and unbiasedness}
\label{app:exact-enumeration}

This appendix records exact finite-sum expressions for the mean and (design-level) variances induced by a product-form QPD, and clarifies unbiasedness for naïve and stratified estimators. These expressions provide ground truth for small instances where explicit enumeration is feasible.

\subsection{QPD ensemble and target mean}

Let a target channel admit a circuit-level QPD
\[
\mathcal U_{\mathrm{circ}}
= \sum_{\underline\ell\in\mathbf I} g(\underline\ell)\,\mathcal U(\underline\ell),
\quad
\|g\|_1 := \sum_{\underline\ell\in\mathbf I} |g(\underline\ell)|,
\]
where $\underline\ell=(\ell_1,\dots,\ell_\nu)$ indexes a configuration in the product space $\mathbf I$ of local primitives.
Fix an input state $\rho_0$ and a Pauli observable $O$ with eigenvalues in $\{\pm1\}$.
Define the configuration-level oracle mean
\[
\mu_{\underline\ell}
:= \Tr\!\big[O\,\mathcal U(\underline\ell)(\rho_0)\big]
= \mathbb E[Z\mid \underline\ell],
\]
where $Z\in\{\pm1\}$ denotes a single measurement outcome when running configuration $\underline\ell$ (Born-rule randomness).
By linearity,
\begin{equation}
\mu
:= \Tr\!\big[O\,\mathcal U_{\mathrm{circ}}(\rho_0)\big]
= \sum_{\underline\ell} g(\underline\ell)\,\mu_{\underline\ell}.
\label{eq:mu_exact_sum}
\end{equation}

\subsection{Naïve QPD sampling and unbiasedness}

The standard QPD sampling distribution is
\begin{equation}
p(\underline\ell) := \frac{|g(\underline\ell)|}{\|g\|_1}.
\label{eq:p_naive}
\end{equation}
Given a sampled configuration $\underline\ell\sim p$ and a measurement outcome $Z$ (possibly compressed after $R$ repeated shots), define the single-configuration estimator
\begin{equation}
Y := \|g\|_1\,\sign(g(\underline\ell))\,Z.
\label{eq:Y_def}
\end{equation}
Then
\[
\mathbb E[Y\mid \underline\ell]
= \|g\|_1\,\sign(g(\underline\ell))\,\mu_{\underline\ell},
\]
and taking expectation over $\underline\ell\sim p$ gives
\begin{align}
\mathbb E[Y]
&= \sum_{\underline\ell} p(\underline\ell)\,\mathbb E[Y\mid \underline\ell]\nonumber\\
&= \sum_{\underline\ell}\frac{|g(\underline\ell)|}{\|g\|_1}\,
\|g\|_1\,\sign(g(\underline\ell))\,\mu_{\underline\ell}\nonumber\\
&= \sum_{\underline\ell} g(\underline\ell)\,\mu_{\underline\ell}
= \mu.
\label{eq:naive_unbiased}
\end{align}
Thus the naïve QPD estimator is unbiased.

\subsection{Stratification and unbiasedness}

Let $S=S(\underline\ell)$ be a deterministic stratification statistic with strata indexed by $s\in\mathcal S$.
For example, the counts-vector statistic $\mathbf{M}$ from the main text. Define stratum weights under the sampling distribution $p$:
\[
w_s := \Pr(S=s)
= \sum_{\underline\ell:\,S(\underline\ell)=s} p(\underline\ell)
= \frac{1}{\|g\|_1}\sum_{\underline\ell\in s} |g(\underline\ell)|.
\]
The conditional distribution within a stratum is
\[
p(\underline\ell\mid s)
= \frac{p(\underline\ell)}{w_s}
= \frac{|g(\underline\ell)|}{\sum_{\underline\ell'\in s}|g(\underline\ell')|}.
\]
Define the stratum conditional mean of $Y$ (under $p$) by
\[
\mu_s := \mathbb E[Y\mid S=s]
= \sum_{\underline\ell\in s} p(\underline\ell\mid s)\,\mathbb E[Y\mid \underline\ell].
\]
Substituting~\eqref{eq:Y_def} yields the convenient closed form
\begin{equation}
\mu_s
= \frac{1}{w_s}\sum_{\underline\ell\in s} g(\underline\ell)\,\mu_{\underline\ell}.
\label{eq:mu_s_exact}
\end{equation}
Therefore $\sum_s w_s\mu_s=\mu$ by summing over all partitions. This can also be obtained by using the law of total expectation.

\paragraph{Unbiased stratified estimator.}
If one draws $K_s\ge 1$ samples within each stratum with $w_s>0$ and sets $\widehat\mu_s$ to the sample mean within stratum $s$, then
\[
\widehat\mu_{\mathrm{strat}} := \sum_s w_s\,\widehat\mu_s
\quad\text{satisfies}\quad
\mathbb E[\widehat\mu_{\mathrm{strat}}]=\mu.
\]
If some positive-mass strata receive $K_s=0$, unbiasedness is preserved by the residual-bucket construction described in Appendix~\ref{app:allocation details} (by sampling from the conditional mixture over the dropped set and adding a residual term $w_*\widehat\mu_*$).

\subsection{Exact second moments and design variances}

This section records exact expressions for the \emph{single-sample/design variance} $\Var(Y)$ under a specified measurement model. The variance of the Monte Carlo estimator with $K$ i.i.d.\ configuration draws is $\Var(\widehat\mu)=\Var(Y)/K$ in the naïve case, and $\Var(\widehat\mu)=\sum_s w_s^2\sigma_s^2/K_s$ (plus a residual term if present) in the stratified case.

\subsubsection*{Oracle model}

In the oracle model we replace the measurement outcome by its conditional mean, i.e.\ $Z\leftarrow \mu_{\underline\ell}$, so $Y$ becomes deterministic given $\underline\ell$ (no shot noise):
\[
Y(\underline\ell) = \|g\|_1\,\sign(g(\underline\ell))\,\mu_{\underline\ell}.
\]
The second moment under $p$ is
\begin{align}
\mathbb E[Y^2]
&= \sum_{\underline\ell} p(\underline\ell)\,Y(\underline\ell)^2
= \sum_{\underline\ell}\frac{|g(\underline\ell)|}{\|g\|_1}\,
\|g\|_1^2\,\mu_{\underline\ell}^2 \nonumber\\
&= \|g\|_1\sum_{\underline\ell} |g(\underline\ell)|\,\mu_{\underline\ell}^2.
\label{eq:E_Y2_oracle}
\end{align}
Hence the naïve oracle design variance is
\begin{equation}
\Var_{\mathrm{oracle}}(Y^{\mathrm{naive}})
= \|g\|_1\sum_{\underline\ell} |g(\underline\ell)|\,\mu_{\underline\ell}^2 - \mu^2.
\label{eq:Var_naive_oracle_exact}
\end{equation}

\paragraph{Stratum moments (oracle).}
Within stratum $s$, $Y$ takes values $Y(\underline\ell)$ with probabilities $p(\underline\ell\mid s)$, so
\[
\sigma_s^2 := \Var(Y\mid S=s)
= \mathbb E[Y^2\mid S=s]-\mu_s^2.
\]
Using $Y(\underline\ell)^2=\|g\|_1^2\mu_{\underline\ell}^2$ and the definition of $p(\underline\ell\mid s)$,
\begin{equation}
\mathbb E[Y^2\mid S=s]
= \frac{\|g\|_1}{w_s}\sum_{\underline\ell\in s} |g(\underline\ell)|\,\mu_{\underline\ell}^2,
\quad
\mu_s=\frac{1}{w_s}\sum_{\underline\ell\in s} g(\underline\ell)\,\mu_{\underline\ell}.
\label{eq:oracle_stratum_moments}
\end{equation}
The proportional stratified \emph{single-sample} design variance is then
\[
\Var_{\mathrm{oracle}}(Y^{\mathrm{prop}})
= \sum_s w_s\,\sigma_s^2,
\]
and the corresponding estimator variance under ideal proportional allocation is $\Var(\widehat\mu^{\mathrm{prop}})=\Var(Y^{\mathrm{prop}})/K$.

\subsubsection*{Single-shot Pauli measurement model}

For Pauli observables with outcomes $Z\in\{\pm1\}$, $\Var(Z\mid \underline\ell)=1-\mu_{\underline\ell}^2$.
Conditioned on $\underline\ell$, the estimator in~\eqref{eq:Y_def} satisfies
\[
\Var(Y\mid \underline\ell)
= \|g\|_1^2\,\Var(Z\mid \underline\ell)
= \|g\|_1^2(1-\mu_{\underline\ell}^2).
\]
Applying the law of total variance over $\underline\ell\sim p$ gives
\begin{align}
\Var(Y)
&= \mathbb E_{\underline\ell}\!\big[\Var(Y\mid \underline\ell)\big]
+ \Var_{\underline\ell}\!\big(\mathbb E[Y\mid \underline\ell]\big)\nonumber\\
&= \|g\|_1^2\sum_{\underline\ell}p(\underline\ell)(1-\mu_{\underline\ell}^2)
+ \Big(\|g\|_1\sum_{\underline\ell}|g(\underline\ell)|\,\mu_{\underline\ell}^2-\mu^2\Big)\nonumber\\
&= \|g\|_1^2-\mu^2.
\label{eq:Var_naive_singleshot_exact}
\end{align}
Thus in Pauli $\pm1$ special case the single-shot model the naïve design variance depends only on $\|g\|_1$ and $\mu$.

\paragraph{Intermediate $R$.}
If each configuration is measured $R$ times and averaged before applying the QPD weight, then $\Var(Z\mid\underline\ell)$ is replaced by $\Var(Z\mid\underline\ell)/R$, so the Born-rule contribution interpolates between~\eqref{eq:Var_naive_singleshot_exact} and~\eqref{eq:Var_naive_oracle_exact}.

\subsection{Use in the main numerics}

The expressions above define exact ground truth quantities for small instances where the configuration space can be enumerated. This was done in Appendix~\ref{app:numerical-methodology} for PEC on a small Trotter instance. In the main simulations we estimate variances from sampled outcomes, and for stratified designs we use the residual-aware allocation and estimator-variance formulas described in Appendix~\ref{app:allocation details}.

\section{Integer allocation, truncation bias, and the residual-stratum fix}
\label{app:allocation details}

This appendix specifies the finite-$K$ allocation used in our experiments and its implications.
Ideal proportional allocation assigns quota $q_s\equiv K w_s$ to each stratum $s$, which is generally non-integer. Since in practice we can only sample an integer number of times, to ensure unbiasedness we must carefully design our sampling scheme.
In this section we (i) recall Hamilton (largest-remainder) apportionment as a deterministic rounding rule, (ii) quantify the truncation bias that results if strata with $K_s=0$ are dropped following a straightforward Hamilton apportionment scheme, and (iii) present a residual-stratum modification that restores exact unbiasedness for any finite $K$.

The results in this appendix section hold for arbitrary stratification statistics $S$, as any proportional allocation scheme necessarily needs to deal with rounding. In the main text and Algorithm~\ref{alg:stratified-qpdmc} we exclusively focus on the counts-vector statistic $\mathbf{M}$ for simplicity.
\subsection{Setup}
Let $S=S(\underline\ell)$ be the stratum statistic, with stratum weights $w_s:=\Pr(S=s)$, $\sum_s w_s=1$.
Let $Y$ denote the per-configuration random variable (including any within-configuration averaging over $R$, if applicable), with
\[
\mu_s := \mathbb{E}[Y\mid S=s],\,\sigma_s^2 := \Var(Y\mid S=s)
\]
such that  $\mu=\sum_s w_s\mu_s$. We use the boundedness assumption
\begin{equation}
|Y|\le B,\quad B:=\ \lVert O\rVert_{\infty}\,\|g\|_1,
\label{eq:boundedY}
\end{equation}
which holds for Hermitian observables with $|O_x|\le \ \lVert O\rVert_{\infty}$ and QPD weights bounded by $\|g\|_1$.

\subsection{Hamilton (largest remainder) apportionment}
\label{app:hamilton}

Hamilton's method maps quotas $q_s=K w_s$ to an array of non-negative integers $(K_s)_s$ with $\sum_s K_s=K$ by rounding down and distributing the remaining units to the largest fractional remainders \cite{balinski_quota_1975, pukelsheim_quota_2017}. In the following we will use $(\cdot)_{s=1}^S$ to denote the array indexed from $s=1$ up to $s=S$. The pseudocode for Hamilton Apportionment is given in Algorithm~\ref{alg:hamilton}.

\begin{algorithm}[H]
\DontPrintSemicolon
\caption{\textsc{HamiltonApportion}$(w,K)$}
\label{alg:hamilton}
\KwIn{$w=(w_s)_{s=1}^S$ with $w_s\ge 0$, $\sum_s w_s=1$; $K\in\mathbb{Z}_{\ge 1}$.}
\KwOut{$(K_s)_{s=1}^S\in\mathbb{Z}_{\ge 0}$ with $\sum_s K_s = K$.}
$q_s \leftarrow K w_s$ for all $s$\;
$K_s \leftarrow \lfloor q_s\rfloor$ for all $s$\;
$r \leftarrow K - \sum_s K_s$\;
$\delta_s \leftarrow q_s - \lfloor q_s\rfloor$ for all $s$\;
Let $I$ be the indices of the $r$ largest $\delta_s$ (ties broken deterministically)\;
\For{$s\in I$}{ $K_s \leftarrow K_s + 1$ }
\Return{$(K_s)_{s=1}^S$}\;
\end{algorithm}

Hamilton apportionment closely matches the quota in the sense that the output allocation satisfies $|K_s - Kw_s| < 1$ for all $s$ and also preserves total budget $\sum_s K_s = K$. However, it can assign $K_s=0$ to strata with $w_s>0$ (whenever $w_s<1/K$). If such strata are omitted from estimation, the resulting estimator is biased for finite $K$, as this is mathematically equivalent to truncating all strata with $w_s < 1/K$ .

\subsection{Bias from dropping zero-allocation strata}
To see how this bias arises, let $\mathcal A:=\{s:K_s>0\}$ and $\mathcal D:=\{s:K_s=0,\ w_s>0\}$ for a given integer allocation, and let
\[
w_{\mathrm{drop}}:=\sum_{s\in\mathcal D} w_s .
\]
If one defines the truncated estimator
\begin{equation}
\widehat\mu_{\mathrm{drop}} := \sum_{s\in\mathcal A} w_s\,\widehat\mu_s,
\label{eq:dropest}
\end{equation}
then
\begin{equation}
\mathrm{Bias}(\widehat\mu_{\mathrm{drop}})
= \mathbb{E}[\widehat\mu_{\mathrm{drop}}]-\mu
= -\sum_{s\in\mathcal D} w_s\mu_s,
\label{eq:bias_exact}
\end{equation}
and by \eqref{eq:boundedY} we obtain the worst-case bound on the bias due to truncation
\begin{equation}
\bigl|\mathrm{Bias}(\widehat\mu_{\mathrm{drop}})\bigr|
\le B\,w_{\mathrm{drop}}
= \ \lVert O\rVert_{\infty}\,\|g\|_1\,w_{\mathrm{drop}}.
\label{eq:bias_bound}
\end{equation}
So when the QPD distribution has a long tail with many strata with $w_s < 1/K$ (or for very small $K$) $w_{\mathrm{drop}}$ can be large, leading to a potentially significant bias if we use Hamilton apportionment directly. This motivates explicitly sampling the dropped mass via a residual stratum to correct the bias, which we implement in Algorithm~\ref{alg:hamilton-residual} \textsc{ResidualHamiltonAllocate}.

\subsection{Residual-stratum construction for exact unbiasedness}
\label{app:allocation-residual}

We introduce a residual stratum $*$ that aggregates all zero-allocation strata, and allocate it an explicit budget $K_*$ while preserving the global budget constraint. This restores exact unbiasedness for finite values of $K$.

\paragraph{Residual stratum.}
Given a dropped set $\mathcal D$ from processing the output of \textsc{HamiltonApportion}, define
\[
w_* := \sum_{s\in\mathcal D} w_s,\quad
q(s):=\Pr(S=s\mid S\in\mathcal D)=\frac{w_s}{w_*}\ \ \ (s\in\mathcal D).
\]
To sample from the stratum means in the residual we draw $s\sim q(\cdot)$, then draw $\underline\ell\sim p(\underline\ell\mid S=s)$, and evaluate $Y$ on this circuit variant.

\paragraph{Allocation routine (Hamilton + residual).}
Our implementation \textsc{ResidualHamiltonAllocate} is specified in Algorithm~\ref{alg:hamilton-residual}. It (i) computes a Hamilton allocation with \textsc{HamiltonApportion}, (ii) plans a residual budget based on the initially dropped mass, (iii) borrows that many units from donor strata to keep the total budget fixed, and (iv) recomputes the dropped set (and hence the residual mixture) after borrowing. We then construct the proportional-stratified estimator..

\paragraph{Estimator construction.}
For each retained stratum $s\in\mathcal A:=\{s:K_s>0\}$, draw $K_s$ i.i.d.\ samples from $p(\underline\ell\mid S=s)$ and form $\widehat\mu_s$.
For the residual stratum, draw $K_*$ samples via the residual mixture and form $\widehat\mu_*$.
Define
\begin{equation}
\widehat\mu_{\mathrm{impl}}
:= \sum_{s\in\mathcal A} w_s\,\widehat\mu_s + w_*\,\widehat\mu_*.
\label{eq:res_est}
\end{equation}

\begin{proposition}[Exact unbiasedness with a residual stratum]
\label{prop:exact-unbiased-residual}
For any finite $K$, any allocations with $K_s\ge 1$ for $s\in\mathcal A$ and $K_*\ge 1$ whenever $w_*>0$, the estimator \eqref{eq:res_est} satisfies $\mathbb{E}[\widehat\mu_{\mathrm{impl}}]=\mu$.
\end{proposition}

\begin{proof}
For $s\in\mathcal A$, $\mathbb E[\widehat\mu_s]=\mu_s$.
For the residual bucket, sampling is from the conditional mixture $S\in\mathcal D$, hence
$\mathbb E[\widehat\mu_*]=\sum_{s\in\mathcal D}(w_s/w_*)\mu_s$.
Substituting into \eqref{eq:res_est} gives $\mathbb{E}[\widehat\mu_{\mathrm{impl}}]=\sum_s w_s\mu_s=\mu$.
\end{proof}

\begin{lemma}[Borrowing cannot exhaust donors]
\label{lem:no-donor-exhaustion}
Assume \textsc{HamiltonApportion} returns $(K_s)_{s=1}^S$ with $\sum_s K_s=K$, and the residual plan satisfies $K_*\le K$.
Then the borrowing phase in Algorithm~\ref{alg:hamilton-residual} always achieves $\mathrm{deficit}=0$ (possibly creating additional zeros in the second pass, which we then re-group into the dropped set afterwards).
\end{lemma}

\begin{proof}
In the second pass, a decrement is permitted whenever the current donor has $K_s\ge 1$.
The total number of admissible decrements before all donors reach zero equals $\sum_s K_s = K$.
Since $K_*\le K$, it is always possible to perform $K_*$ decrements in total; hence $\mathrm{deficit}$ must reach zero.
\end{proof}

\paragraph{Remark (why we recompute $\mathcal D$).}
Since borrowing in the second pass may create new zero-allocated strata, to preserve exact unbiasedness, the residual mixture must be defined using the \emph{final} dropped set
$\mathcal D=\{s:K_s=0,\ w_s>0\}$ after borrowing, as in Algorithm~\ref{alg:hamilton-residual}.

\subsection{Variance impact of rounding and the residual fix}
\label{app:var-rounding-residual}

Under ideal proportional allocation $K_s=K w_s$ (non-integer in general), the proportional stratified estimator has the simple variance result (see also \eqref{eq:var-strat-general-QPD}) 
\begin{equation}
\Var(\widehat\mu_{\mathrm{prop}})
=\sum_s \frac{w_s^2\sigma_s^2}{K w_s}
= \frac{1}{K}\sum_s w_s\sigma_s^2.
\label{eq:var-prop-ideal}
\end{equation}
With integer allocations and a residual stratum, the implemented estimator variance is
\begin{equation}
\Var(\widehat\mu_{\mathrm{impl}})
= \sum_{s\in\mathcal{A}} \frac{w_s^2\sigma_s^2}{K_s}
+ \frac{w_*^2\sigma_*^2}{K_*},
\label{eq:var-impl}
\end{equation}
where $\sigma_*^2:=\Var(Y\mid S\in\mathcal{D})$ and $\mathcal A,\mathcal D$ are defined by the final allocation.

\paragraph{Residual variance decomposition.}
Under the residual mixture $q(s)=w_s/w_*$, the law of total variance gives
\begin{equation}
\sigma_*^2
=\sum_{s\in\mathcal{D}} q(s)\,\sigma_s^2
+\Var_{q}(\mu_s).
\label{eq:sigma-star-decomp}
\end{equation}

\begin{algorithm}[H]
\DontPrintSemicolon
\caption{\textsc{ResidualHamiltonAllocate}$(w,K)$}
\label{alg:hamilton-residual}
\KwIn{$w=(w_s)_{s=1}^S$ with $w_s\ge 0$, $\sum_s w_s=1$; $K\in\mathbb{Z}_{\ge 1}$.}
\KwOut{Allocations $(K_s)_{s=1}^S$, residual count $K_*$, final dropped set $\mathcal D$, dropped mass $w_*$, residual mixture $q(\cdot)$.}
$(K_s)_{s=1}^S \leftarrow \textsc{HamiltonApportion}(w,K)$\;
$\mathcal D_0 \leftarrow \{s:K_s=0,\ w_s>0\}$; \quad $w_{\mathrm{drop},0}\leftarrow \sum_{s\in\mathcal D_0} w_s$\;
$K_* \leftarrow 0$\;
\If{$w_{\mathrm{drop},0}>0$}{
    $K_* \leftarrow \max\{1,\mathrm{round}(K\,w_{\mathrm{drop},0})\}$\;
    $\mathrm{deficit}\leftarrow K_*$\;
    Order strata by decreasing current $K_s$\;
    \ForEach{$s$ in order}{
        \While{$\mathrm{deficit}>0$ \KwSty{and} $K_s\ge 2$}{
            $K_s\leftarrow K_s-1$;\ \ $\mathrm{deficit}\leftarrow \mathrm{deficit}-1$\;
        }
        \If{$\mathrm{deficit}=0$}{\textbf{break}}
    }
    \If{$\mathrm{deficit}>0$}{
        Re-order strata by decreasing current $K_s$\;
        \ForEach{$s$ in order}{
            \While{$\mathrm{deficit}>0$ \KwSty{and} $K_s\ge 1$}{
                $K_s\leftarrow K_s-1$;\ \ $\mathrm{deficit}\leftarrow \mathrm{deficit}-1$\;
            }
            \If{$\mathrm{deficit}=0$}{\textbf{break}}
        }
    }
    \tcp{Donors cannot be exhausted (Lemma~\ref{lem:no-donor-exhaustion}).}
    \textbf{assert} $\mathrm{deficit}=0$\;
}
$\mathcal D \leftarrow \{s:K_s=0,\ w_s>0\}$;\quad $w_* \leftarrow \sum_{s\in\mathcal D} w_s$\;
\If{$K_*>0$ \KwSty{and} $w_*>0$}{
    Define $q(s)=w_s/w_*$ for $s\in\mathcal D$\;
}
\Else{
    $K_*\leftarrow 0$;\quad $\mathcal D\leftarrow \emptyset$;\quad $w_*\leftarrow 0$\;
}
\textbf{assert} $\sum_{s=1}^S K_s + K_* = K$\;
\Return $(K_s)_{s=1}^S,\ K_*,\ \mathcal D,\ w_*,\ q(\cdot)$\;
\end{algorithm}

\paragraph{Exact perturbation identity.}
Combining \eqref{eq:var-prop-ideal}--\eqref{eq:sigma-star-decomp} yields
\begin{align}
\Var(\widehat\mu_{\mathrm{impl}})-&\Var(\widehat\mu_{\mathrm{prop}})
=
\sum_{s\in\mathcal{A}} w_s^2\sigma_s^2\Bigl(\frac{1}{K_s}-\frac{1}{K w_s}\Bigr)
\nonumber\\
&+
\Bigl(\frac{w_*}{K_*}-\frac{1}{K}\Bigr)\sum_{s\in\mathcal{D}} w_s\sigma_s^2
+
\frac{w_*^2}{K_*}\Var_{q}(\mu_s).
\label{eq:var-perturbation-decomp}
\end{align}

\paragraph{A single computable certificate.}
Using \eqref{eq:boundedY}, we have $\sigma_s^2\le B^2$ and $\Var_q(\mu_s)\le B^2$, hence
\begin{equation}
\bigl|\Var(\widehat\mu_{\mathrm{impl}})-\Var(\widehat\mu_{\mathrm{prop}})\bigr|
\le
\mathsf{Cert}_{\mathrm{var}},
\label{eq:cert-var-def}
\end{equation}
where
\begin{equation}
\mathsf{Cert}_{\mathrm{var}}
:=
B^2\Bigg[
\sum_{s\in\mathcal{A}} w_s^2\Bigl|\frac{1}{K_s}-\frac{1}{K w_s}\Bigr|
\;+\;
w_*\Bigl|\frac{w_*}{K_*}-\frac{1}{K}\Bigr|
\;+\;
\frac{w_*^2}{K_*}
\Bigg].
\label{eq:cert-var}
\end{equation}
Equivalently, $\sqrt{\mathsf{Cert}_{\mathrm{var}}}$ upper-bounds the possible perturbation in estimator standard error induced purely by integer apportionment and residual aggregation.

\paragraph{Typical scaling.}
If $K_s\approx K w_s$ for $s\in\mathcal A$ and $K_*\approx K w_*$, then the certificate terms are small; in particular, when $K_*\asymp K w_*$ the residual contributions scale as
\[
w_*\Bigl|\frac{w_*}{K_*}-\frac{1}{K}\Bigr| = O\!\left(\frac{w_*}{K^2}\right),
\quad
\frac{w_*^2}{K_*} = O\!\left(\frac{w_*}{K}\right),
\]
so small final dropped mass $w_*$ suppresses the residual effect relative to the baseline $O(1/K)$ Monte Carlo scaling.

We will use the certificate values in Appendix~\ref{app:numerical-methodology} to verify that the implemented sampling routine under residual Hamilton apportionment does not significantly deviate from the ideal (non-integer) proportional allocation scheme. Specifically, we find that the certificate values obtained in our numerics are substantially smaller than the error bars, suggesting that the integerisation process do not materially affect our conclusions as to the variance reduction that proportional stratified sampling achieves.

\section{Dynamic programmes for counts-vector strata}
\label{app:counts-dp}

This appendix records the dynamic programmes used to (i) compute counts-vector stratum weights
$w_{\mathbf m}=\Pr(\mathbf M=\mathbf m)$ and (ii) sample configurations $\underline{\ell}$ conditionally on a target counts
vector $\mathbf m$. Throughout, we focus on product distributions of the form
\[
p(\underline\ell)=\prod_{i=1}^{\nu} p_i(\ell_i),
\quad \ell_i\in[d]:=\{1,\dots,d\},
\]
as induced by product-form QPD sampling with individual-primtiive probabilities $p_i(\ell_i)\propto |\gamma_i(\ell_i)|$ (after converting signed coefficients into probabilities via $\|g\|_1$). The statistic is the counts vector $\mathbf M(\underline\ell)$ in~\eqref{eq:counts-vector}.
The induced distribution of $\mathbf M$ is a Poisson-multinomial distribution (PMD), i.e.\ a sum of
independent (but not necessarily identically distributed) categorical variables. 

\subsection{Forward DP for stratum weights}
For $i\in\{0,1,\dots,\nu\}$ and any $\mathbf m\in\mathbb Z_{\ge 0}^d$ with $\sum_{k} m_k=i$, define
\[
W^{(i)}_{\mathbf m}
:=\Pr\!\bigl(\mathbf M^{(i)}=\mathbf m\bigr),
\]
where $\mathbf M^{(i)}$ is the counts vector of the prefix vector $(\ell_1,\dots,\ell_i)$.
Initialise $W^{(0)}_{\mathbf 0}=1$ and $W^{(0)}_{\mathbf m\neq \mathbf 0}=0$.
The standard PMD recursion is
\begin{equation}
W^{(i)}_{\mathbf m}
=
\sum_{k=1}^d p_i(k)\,W^{(i-1)}_{\mathbf m-e_k},
\label{eq:forward-dp-rec}
\end{equation}
where $e_k$ is the $k$th unit vector and $W^{(i-1)}_{\mathbf m-e_k}$ is interpreted as $0$ if any component
of $\mathbf m-e_k$ is negative. The desired stratum weights are the final layer
\[
w_{\mathbf m}:=W^{(\nu)}_{\mathbf m}
\quad(\sum_k m_k=\nu).
\]
At layer $i$, the number of admissible states is $\binom{i+d-1}{d-1}$; hence for fixed $d$, the total work needed to traverse all layers is $O(d\,\nu^{d})$. Storing all layers uses $O(\nu^{d})$ memory and is convenient for the backward sampler
below. (If only $w_{\mathbf m}$ are needed, a rolling two-layer implementation suffices; however in our implementation we always store the full table). Note also that if we store $\mathbf m$ as a tuple of length $d-1$ (since the last coordinate is determined by $\sum m_k=i$), we can reduce memory by one row; however this means that look-ups become correspondingly more expensive and we lose the ability to sanity check each intermediate calculation). For simplicity we did not implement this optimisation.

\subsection{Backward conditional sampling given a counts vector}
Fix a target $\mathbf m$ with $\sum_k m_k=\nu$ and $w_{\mathbf m}>0$. We generate
$\underline\ell\sim p(\underline\ell\mid \mathbf M=\mathbf m)$ by sampling indices in reverse order using
the cached table $\mathcal T$ and Bayes' Rule. The details are in Algorithm~\ref{alg:counts-conditional}.

Specifically, let $\mathbf m^{(\nu)}=\mathbf m$. For $i=\nu,\nu-1,\dots,1$, define
\begin{align} 
q_i(k\mid \mathbf m^{(i)})
&:=\Pr(\ell_i=k\mid \mathbf M^{(i)}=\mathbf m^{(i)})
\nonumber\\
&= \frac{p_i(k)\,W^{(i-1)}_{\mathbf m^{(i)}-e_k}}
{W^{(i)}_{\mathbf m^{(i)}}},
\label{eq:backward-prob}
\end{align}
with the convention $q_i(k\mid\mathbf m^{(i)})=0$ when $\mathbf m^{(i)}-e_k$ is invalid (any negative integer entries). Sample
$\ell_i\sim q_i(\cdot\mid \mathbf m^{(i)})$ and update
$\mathbf m^{(i-1)}\leftarrow \mathbf m^{(i)}-e_{\ell_i}$.
The resulting $\underline\ell=(\ell_1,\dots,\ell_\nu)$ has counts vector $\mathbf m$ by construction. Since this only requires traversing all rows in the table once, the time-complexity of each conditional sample is just $O(\nu d)$.  

\begin{algorithm}[H]
\caption{\textsc{CountsForwardDP}$(\{p_i\}_{i=1}^{\nu})$}
\label{alg:counts-forward}
\KwIn{Local categorical probabilities $p_i(\cdot)$ on $[d]$ for each $i\in[\nu]$.}
\KwOut{Final weights $\{w_{\mathbf m}\}$ and cached table $\mathcal T=\{W^{(i)}_{\mathbf m}\}_{i=0}^{\nu}$.}

Initialise $W^{(0)}_{\mathbf 0}\leftarrow 1$ and $W^{(0)}_{\mathbf m\neq \mathbf 0}\leftarrow 0$\;
\For{$i\leftarrow 1$ \KwTo $\nu$}{
  \ForEach{$\mathbf m\in\mathbb Z_{\ge 0}^d$ with $\sum_k m_k=i$}{
    $W^{(i)}_{\mathbf m}\leftarrow 0$\;
    \For{$k\leftarrow 1$ \KwTo $d$}{
      \If{$m_k>0$}{
        $W^{(i)}_{\mathbf m}\leftarrow W^{(i)}_{\mathbf m}+p_i(k)\,W^{(i-1)}_{\mathbf m-e_k}$\;
      }
    }
  }
}
Set $w_{\mathbf m}\leftarrow W^{(\nu)}_{\mathbf m}$ for all $\mathbf m$ with $\sum_k m_k=\nu$\;
\Return $(\{w_{\mathbf m}\},\ \mathcal T)$\;
\end{algorithm}

\subsection{Correctness (sketch)}
We briefly justify that (i) the forward DP computes the exact PMD prefix probabilities, and (ii) the
backward procedure produces exact conditional samples.

\begin{proposition}[Forward DP correctness]
For every $i\in\{0,1,\dots,\nu\}$ and every $\mathbf m\in\mathbb Z_{\ge 0}^d$ with $\sum_k m_k=i$,
the forward DP values satisfy
\[
W^{(i)}_{\mathbf m} \;=\; \Pr\!\bigl(\mathbf M^{(i)}=\mathbf m\bigr)
\quad\text{under}\quad
p(\ell_{1:i})=\prod_{j=1}^{i} p_j(\ell_j).
\]
In particular, the stratum weights are $w_{\mathbf m}=W^{(\nu)}_{\mathbf m}$.
\end{proposition}

\begin{proof}
We argue by induction on $i$. The base case $i=0$ holds by construction:
$W^{(0)}_{\mathbf 0}=1=\Pr(\mathbf M^{(0)}=\mathbf 0)$ and $W^{(0)}_{\mathbf m}=0$ otherwise.
For the inductive step, fix $i\ge 1$ and an admissible $\mathbf m$ with $\sum_k m_k=i$.
Partition on the last index $\ell_i\in[d]$, then:
\[
\Pr(\mathbf M^{(i)}=\mathbf m)
=\sum_{k=1}^d \Pr(\ell_i=k)\,\Pr(\mathbf M^{(i-1)}=\mathbf m-e_k).
\]
Since $\Pr(\ell_i=k)=p_i(k)$ and $\Pr(\mathbf M^{(i-1)}=\mathbf m-e_k)=0$ when $\mathbf m-e_k$ has
a negative entry, the recursion is exactly~\eqref{eq:forward-dp-rec}. Applying the inductive
hypothesis to the $(i\!-\!1)$-prefix probabilities yields $W^{(i)}_{\mathbf m}=\Pr(\mathbf M^{(i)}=\mathbf m)$. 
\end{proof}

\begin{proposition}[Backward sampler correctness]
Fix any target $\mathbf m$ with $\sum_k m_k=\nu$ and $w_{\mathbf m}=W^{(\nu)}_{\mathbf m}>0$.
Algorithm~\ref{alg:counts-conditional} outputs a configuration
$\underline\ell=(\ell_1,\dots,\ell_\nu)$ distributed as
$p(\underline\ell\mid \mathbf M=\mathbf m)$.
\end{proposition}

\begin{proof}
Let $\mathbf m^{(\nu)}=\mathbf m$, and after sampling $\ell_i$ define the update
$\mathbf m^{(i-1)}=\mathbf m^{(i)}-e_{\ell_i}$ as in the algorithm. For each $i$ and any admissible
$\mathbf m^{(i)}$ with $W^{(i)}_{\mathbf m^{(i)}}>0$, Bayes' rule gives
\begin{align*}
    \Pr(\ell_i=k \mid \mathbf M^{(i)}=\mathbf m^{(i)})
&=\frac{\Pr(\ell_i=k,\mathbf M^{(i)}=\mathbf m^{(i)})}{\Pr(\mathbf M^{(i)}=\mathbf m^{(i)})}\\
&=\frac{p_i(k)\,W^{(i-1)}_{\mathbf m^{(i)}-e_k}}{W^{(i)}_{\mathbf m^{(i)}}}
\end{align*}

which is exactly~\eqref{eq:backward-prob} and matches the categorical law used in
Algorithm~\ref{alg:counts-conditional}.
Applying the chain rule for conditional probabilities yields the joint law of the sampled sequence:
\[
\Pr(\ell_{1:\nu}=\underline\ell \mid \mathbf M=\mathbf m)
=\prod_{i=\nu}^{1} \Pr(\ell_i \mid \mathbf M^{(i)}=\mathbf m^{(i)}).
\]
The update rule $\mathbf m^{(i-1)}=\mathbf m^{(i)}-e_{\ell_i}$ ensures that the sampled prefix
always remains consistent with the remaining counts. Therefore, the algorithm samples exactly from
$p(\underline\ell\mid \mathbf M=\mathbf m)$.
\end{proof}

\begin{algorithm}[H]
\caption{\textsc{CountsConditionalSample}$(\mathbf m,\{p_i\}_{i=1}^{\nu},\mathcal T)$}
\label{alg:counts-conditional}
\KwIn{Target $\mathbf m\in\mathbb Z_{\ge 0}^d$ with $\sum_k m_k=\nu$; local probabilities $\{p_i\}$; cached table $\mathcal T=\{W^{(i)}_{\mathbf m}\}$.}
\KwOut{Configuration $\underline\ell\sim p(\underline\ell\mid \mathbf M=\mathbf m)$.}

$\mathbf m^{(\nu)}\leftarrow \mathbf m$\;
\For{$i\leftarrow \nu$ \KwTo $1$}{
  \For{$k\leftarrow 1$ \KwTo $d$}{
    \eIf{$m^{(i)}_k>0$ and $W^{(i)}_{\mathbf m^{(i)}}>0$}{
      $q_i(k)\leftarrow \dfrac{p_i(k)\,W^{(i-1)}_{\mathbf m^{(i)}-e_k}}{W^{(i)}_{\mathbf m^{(i)}}}$\;
    }{
      $q_i(k)\leftarrow 0$\;
    }
  }
  Sample $\ell_i$ from the categorical distribution $q_i(\cdot)$\;
  $\mathbf m^{(i-1)}\leftarrow \mathbf m^{(i)}-e_{\ell_i}$\;
}
\Return $\underline\ell=(\ell_1,\dots,\ell_\nu)$\;
\end{algorithm}

\paragraph{Intuition.}
The backward pass is a sequence of \emph{exact} conditional draws: at step $i$ it samples $\ell_i$
from the correct conditional distribution given the remaining counts $\mathbf m^{(i)}$, then updates
the remaining counts accordingly. This enforces $\mathbf M(\underline\ell)=\mathbf m$ while preserving the
correct conditional law.

\subsection{Implementation notes and alternatives}

The forward DP described above is an exact algorithm for Poisson--multinomial probabilities; it is a generalisation of the DP for calculating the density function of a Poisson-binomial distribution \cite{hong_computing_2013} and also see the PMD review in Ref.~\cite{lin_poisson_2022}. The Poisson-binomial case can be easily recovered as a special case. 

Alternatively, computing only the PMD marginal weights $\{w_{\mathbf m}\}$ can also be done by polynomial/convolution methods (including FFT/DFT variants) and by Gaussian approximation schemes that scale to substantially larger $(\nu,d)$ \cite{lin_poisson_2022}. Our application, however, requires repeated sampling from the conditional law $p(\underline\ell\mid \mathbf M=\mathbf m)$ across many strata, which these other algorithms do not readily facilitate. However, there are trade-offs to our approach.

\paragraph{Scaling, memory, and parallelism.}
A practical limitation of our implementation of counts-vector stratification is the growth of the DP state space with width $d$. The number of distinct counts vectors for fixed $(\nu,d)$ is $|\mathcal S|=\binom{\nu+d-1}{d-1}=O(\nu^{d-1})$, and caching the intermediate tables required for exact backward sampling scales as $O(\nu^{d})$ in memory. So while the complexity scales polynomially in $\nu$ for a fixed $d$, it scales \emph{exponentially} in $d$. Figure~\ref{fig:dp_state_count} reports the total number of cached DP states $\sum_{i=0}^{\nu}\binom{i+d-1}{d-1}$ as a function of $(\nu,d)$. In our numerical experiments (Section~\ref{sec:numerical results}), moderate depths together with modest widths already approach memory limits on consumer hardware ($\nu \sim 100$, $d=4$). The forward pass costs $O(d\,\nu^{d})$ arithmetic operations, while exact backward sampling requires access to the cached tables $\{W^{(i)}_{\mathbf m}\}$. We note that it is also possible to compute marginal values ``on-the-fly'' rather than storing the full table, at the cost of increasing the run-time required to generate each sample. However, this is fairly inefficient, as once these tables are cached conditional samples $\underline\ell\sim p(\underline\ell\mid \mathbf M=\mathbf m)$ can be generated independently, so the backward sampler can be parallelised (with shared read-only access to the cached DP table).

\paragraph{Numerical stability.}
Because $w_{\mathbf m}$ can become very small for large $\nu$, numerical underflow can arise in principle.
Standard safeguards include storing $W$ in log space, or renormalising each layer and tracking scaling factors.
We did not encounter instabilities in the instances benchmarked here, but guarding against underflow is
important for larger depths and/or more extreme inhomogeneity.

\begin{figure}
    \centering
    \includegraphics[width=\linewidth]{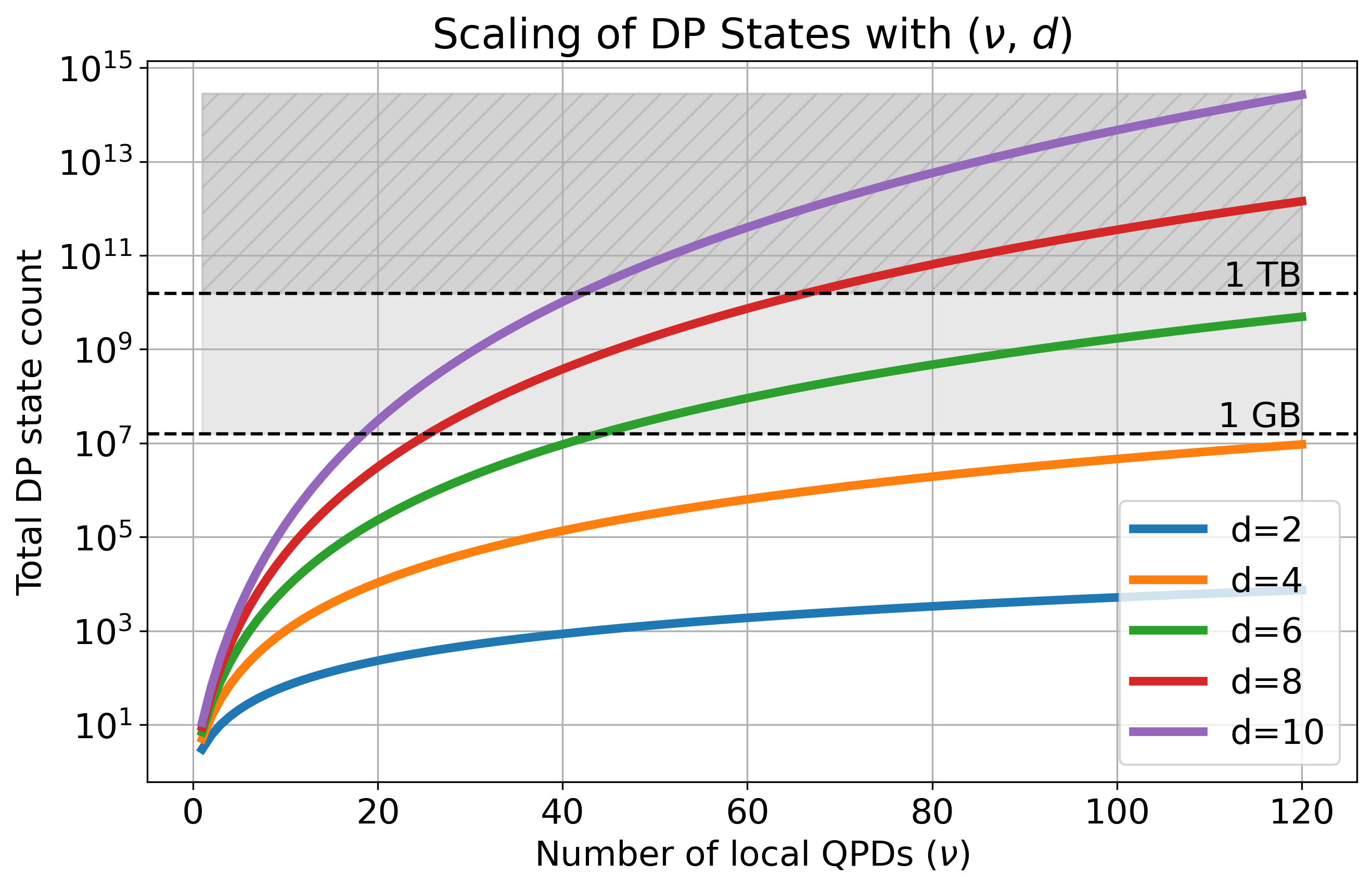}
    \caption{State-space scaling of the cached counts-vector DP. Shown is the number of reachable DP states
    $\sum_{i=0}^{\nu}\binom{i+d-1}{d-1}$, i.e.\ the size of the cached forward table required to compute exact
    counts-marginals and to enable exact backward conditional sampling. As a
    simple architecture-agnostic proxy, we annotate 1\,GB/1\,TB reference lines assuming one float64 (8 bytes)
    stored per state. Real implementations typically incur constant-factor overheads (data type, number of cached layers, metadata, etc.), or benefit from
    sparse/compressed representations, but the state count captures the fundamental scaling with $(\nu,d)$.}
    \label{fig:dp_state_count}
\end{figure}

\paragraph{State-space concentration and motivation for pruning.}
Although the DP must enumerate all reachable counts vectors to compute exact marginals and support exact conditional sampling, the induced Poisson--multinomial distribution is often sharply concentrated in practical product-form QPDs, because certain local primitives are much rarer than others (e.g.\ antipodal rotations in PAI, or non-identity inverse-error insertions in PEC). Figure~\ref{fig:pai_pmd_concentration} quantifies this effect for the TFIM--PAI instances of Section~\ref{sec:numerical results}: for each Trotter depth $L$ we sort strata by decreasing mass and plot the cumulative mass $\sum_{i\le t} w_{(i)}$. Across depths, the smallest index $t_{0.99}$ capturing $99\%$ of the total mass corresponds to only $\sim 2\%$--$9\%$ of the reachable strata. Operationally, this strongly motivates approximate variants in which the DP computation and/or cached state space are pruned toward the high-mass region. However, we do \emph{not} use pruning in the main experiments: all results in the paper use the exact DP to compute $\{w_{\mathbf m}\}$ and to enable exact sampling from $p(\underline\ell\mid \mathbf M=\mathbf m)$. Nevertheless, Fig.~\ref{fig:pai_pmd_concentration} suggests that for larger $(\nu,d)$ one can plausibly trade a controlled approximation (or controlled bias, if desired) for substantial reductions in preprocessing cost. Developing pruning/coarsening rules that preserve a usable backward sampler is therefore a natural next step.

\begin{figure}
    \centering
    \includegraphics[width=\linewidth]{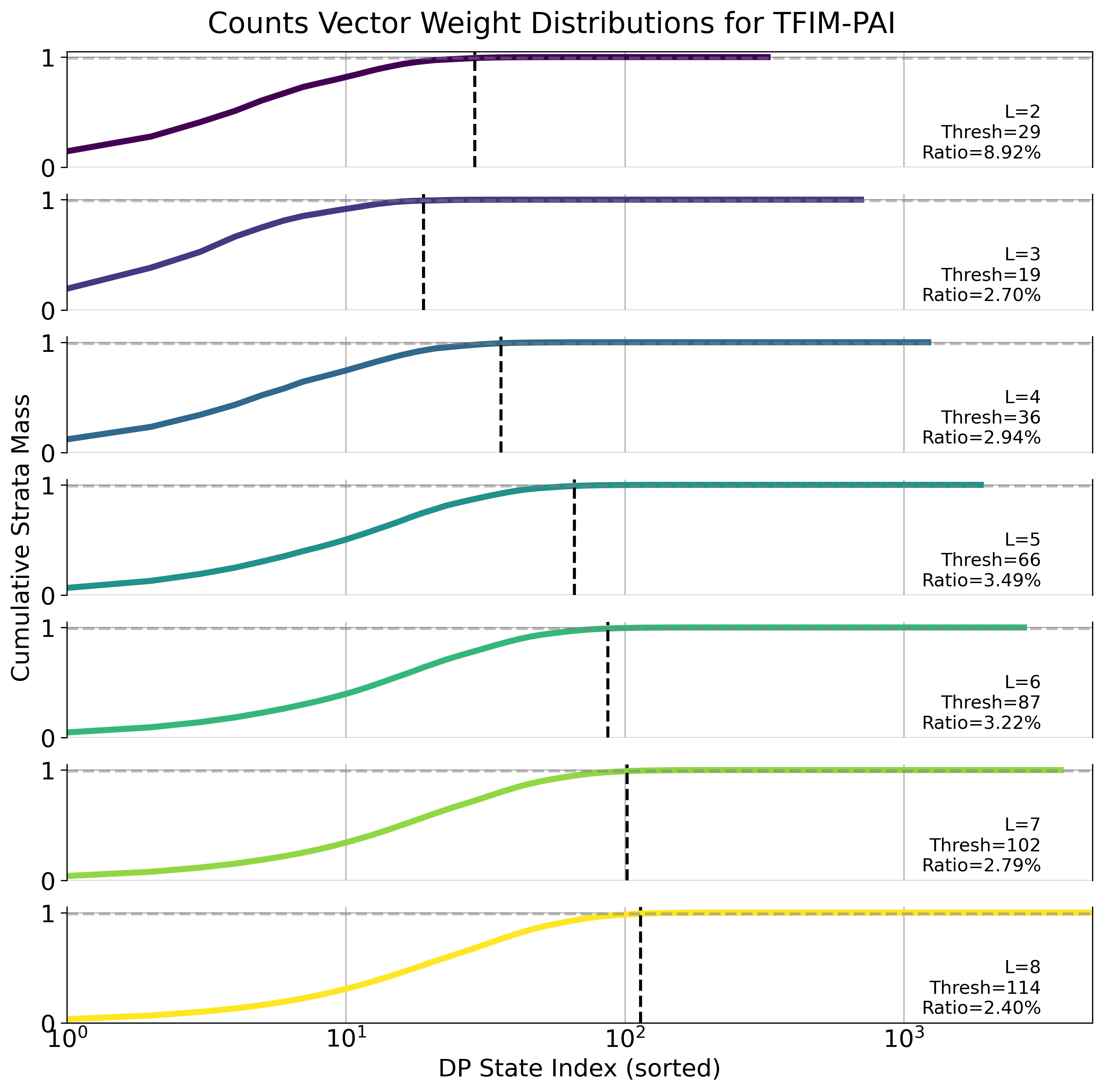}
    \caption{Counts-vector stratum mass concentration for TFIM--PAI.
For each Trotter depth $L$, we compute the counts-vector stratum weights $w_{\mathbf m}=\Pr(\mathbf M=\mathbf m)$ (Poisson--multinomial marginals) and sort strata by decreasing mass $w_{(1)}\ge w_{(2)}\ge\cdots$. Shown is the cumulative mass $\sum_{i\le t} w_{(i)}$ versus sorted index $t$ (log-scale). The horizontal dashed line marks 0.99, and the vertical dashed line marks the smallest index $t_{0.99}$ such that $\sum_{i\le t_{0.99}} w_{(i)}\ge 0.99$. Each panel reports the threshold $t_{0.99}$ and the ratio $t_{0.99}/|\mathcal S|$ (labelled `Thresh' and `Ratio' respectively), where $|\mathcal S|$ is the number of reachable DP states (counts-vector strata).
In this TFIM--PAI instance ($d=3$ at every location), $|\mathcal S|=\binom{\nu+d-1}{d-1}$.
Across depths, $99\%$ of the stratum mass is supported on only $\sim 2\%$--$9\%$ of strata, illustrating a
pronounced head--tail structure and motivating truncated/coarsened variants of the exact DP.}
    \label{fig:pai_pmd_concentration}
\end{figure}

\section{Numerical methodology: variance estimation, bootstrap uncertainty, absolute variances, and validation}
\label{app:numerical-methodology}

This appendix documents the numerical procedures used in Section~\ref{sec:numerical results} to (i) estimate the
\emph{estimator variance} under each sampling design, (ii) attach uncertainty bands via a nonparametric bootstrap,
(iii) report \(\|g\|_1\)-normalised \emph{absolute} variances (Appendix~\ref{subsec:abs-variances}), and (iv) validate
our implementations against an exactly enumerable toy instance (Appendix~\ref{subsec:mc-validation}).

Throughout, a single run of a sampling design (naïve or proportional-stratified) produces realised per-configuration scalars \(\{Y_j\}_{j=1}^K\), where each
\(Y_j\) may already include an \(R\)-shot average as in~\eqref{eq:Yk-QPD}. We write
\[
\hat{\mu}:=\widehat{Y}^{(R)}_K
\]
for the corresponding mean estimator (naïve or stratified) of this dataset. 

\subsection{Estimator variance versus design variance}
\label{subsec:var-vs-designvar}

The main reported quantity is the variance of the \emph{mean estimator} \(\Var(\hat{\mu})\).
For the naïve i.i.d.\ design, \(\Var(\hat{\mu}_{\mathrm{naive}})=\Var(Y)/K\).
For ideal proportional stratification (non-integer quotas) one has
\(
\Var(\hat{\mu}_{\mathrm{prop}})=(1/K)\sum_s w_s\sigma_s^2
\),
while our implemented estimator uses the residual-aware form from Appendix~\ref{app:allocation details} which has an additional correction term to the variance (which we emphasise is small in our experiments).

In figures intended to compare against \(K\)-independent reference values (e.g.\ convergence checks),
we additionally report the \emph{scaled estimator variance}
\[
K\,\widehat{\Var}(\hat{\mu}),
\]
which equals \(\Var(Y)\) for the i.i.d.\ design and equals \(\sum_s w_s\sigma_s^2\) for ideal proportional
stratification. Since this is a linear scaling, all confidence intervals scale directly.

\subsection{Plug-in estimators for the Variance}
\label{subsec:plugin-varmu}
Here we note the formulae for estimating \(\Var(\hat{\mu})\).
\paragraph{Naïve design.}
Given realised samples \(\{Y_j\}_{j=1}^K\), define
\[
\bar{Y}:=\frac{1}{K}\sum_{j=1}^K Y_j,
\quad
s^2 := \frac{1}{K-1}\sum_{j=1}^K (Y_j-\bar{Y})^2.
\]
We use the standard plug-in estimator
\begin{equation}
\widehat{\Var}(\hat{\mu}_{\mathrm{naive}}) := \frac{s^2}{K}.
\label{eq:plugin_naive_varmu_app}
\end{equation}

\paragraph{Residual-aware stratified design.}
Let \(\mathcal{A}\) be the retained strata with realised counts \(K_s\ge 1\), and let \((w_*,K_*)\) denote the
residual bucket weight and realised count (Appendix~\ref{app:allocation details}).
Within each retained stratum compute the unbiased sample variance \(s_s^2\) from the realised
\(\{Y_{s,t}\}_{t=1}^{K_s}\), and similarly compute \(s_*^2\) from residual draws if \(K_*>1\).
We then estimate
\begin{equation}
\widehat{\Var}(\hat{\mu}_{\mathrm{impl}})
:=
\sum_{s\in\mathcal{A}} w_s^2\,\frac{s_s^2}{K_s}
\;+\;
w_*^2\,\frac{s_*^2}{K_*}
\label{eq:plugin_strat_varmu_app}
\end{equation}
using the exact population quantities $w_s$ and $w_*$ obtained via the DP. The resulting estimator of the variance is the empirical quantity used both for absolute variances and for variance ratios in the main text.

\subsection{Nonparametric bootstrap uncertainty bands}
\label{subsec:bootstrap}

To obtain uncertainty bands on our estimate of the variance we use a nonparametric bootstrap with \(B=1024\) resamples in all plots.

\paragraph{Bootstrap for \(\widehat{\Var}(\hat{\mu})\): naïve.}
For \(b=1,\dots,B\), draw a resample \(\{Y_j^{*(b)}\}_{j=1}^K\) by sampling with replacement from
\(\{Y_j\}_{j=1}^K\). Compute \(s_{*(b)}^2\) and set
\[
\widehat{\Var}^{(b)}(\hat{\mu}_{\mathrm{naive}}):=s_{*(b)}^2/K.
\]
Percentile intervals of \(\{\widehat{\Var}^{(b)}(\hat{\mu}_{\mathrm{naive}})\}_{b=1}^B\) define the reported bands.

\paragraph{Bootstrap for \(\widehat{\Var}(\hat{\mu})\): stratified (residual-aware).}
For each stratum \(s\in\mathcal{A}\), resample \emph{within the stratum} by sampling with replacement from the
realised \(\{Y_{s,t}\}_{t=1}^{K_s}\), preserving \(K_s\). If a residual bucket is present, resample within the residual
pool of size \(K_*\). Compute the resampled variances \(s_{s,*(b)}^2\) and (if applicable) \(s_{*,*(b)}^2\), and set
\[
\widehat{\Var}^{(b)}(\hat{\mu}_{\mathrm{impl}})
:=
\sum_{s\in\mathcal{A}} w_s^2\,\frac{s_{s,*(b)}^2}{K_s}
\;+\;
w_*^2\,\frac{s_{*,*(b)}^2}{K_*}.
\]
Percentile intervals over \(b\) again yield uncertainty bands.

\paragraph{Bootstrap for variance ratios.}
For a fixed problem instance and fixed \(K\), let \(\widehat{\Var}(\hat{\mu}_{\mathrm{impl}})\) and
\(\widehat{\Var}(\hat{\mu}_{\mathrm{naive}})\) be the plug-in estimator variances from independent runs. We report
\[
\widehat{\rho}
:=
\frac{\widehat{\Var}(\hat{\mu}_{\mathrm{impl}})}
{\widehat{\Var}(\hat{\mu}_{\mathrm{naive}})}.
\]
To attach bands, we bootstrap the two runs independently and form
\[
\widehat{\rho}^{(b)}
:=
\frac{\widehat{\Var}^{(b)}(\hat{\mu}_{\mathrm{impl}})}
{\widehat{\Var}^{(b)}(\hat{\mu}_{\mathrm{naive}})}.
\]
Percentiles of \(\{\widehat{\rho}^{(b)}\}_{b=1}^B\) define the reported confidence bands. When both designs use the
same \(K\) (as in this paper), the factor of \(K\) cancels in the ratio, so \(\widehat{\rho}\) is identical whether interpreted in terms of
\(\Var(\hat{\mu})\) or \(K\Var(\hat{\mu})\).

\subsection{Absolute variances}
\label{subsec:abs-variances}

Section~\ref{sec:numerical results} emphasised the ratio between the naïve and proportional-stratified estimators which isolated the constant-factor improvement due to stratified configuration sampling. For completeness, here in
Fig.~\ref{fig:absvar-pec} and \ref{fig:absvar-PAI} we report the corresponding \emph{absolute} estimator variances normalised by the
circuit QPD 1-norm squared, i.e. \(\Var(\widehat Y_{K,R})/\|g\|_1^2\), for both naïve and proportional designs to emphasise the identical asymptotic scaling.

For Pauli-valued observables \(|O_x|\le 1\) and QPD weights satisfy \(|w(\underline\ell)|=\|g\|_1\), so
\(|Y^{(R)}|\le \|g\|_1\) almost surely. Hence
\[
\frac{\Var(\widehat Y_{K,R})}{\|g\|_1^2}
= \frac{1}{K}\frac{\Var(Y^{(R)})}{\|g\|_1^2}
\le \frac{1}{K},
\]
independently of depth \(L\) and measurement model \(R\). The grey horizontal line in both plots shows this bound for our fixed \(K=8192\).

Specifically, normalising by \(\|g\|_1^2\) removes the dominant exponential scale set by the QPD overhead, allowing remaining dependence on \(L\) and \(R\) to be compared directly. The variance ratios \(\rho(R)\) in the main text are obtained by taking, at each \((L,R)\), the ratio of the stratified and naïve curves. Normalised absolute estimator variances for the TFIM benchmark. We plot \(\widehat{\Var}(\widehat Y_{K,R})/\|g\|_1^2\) versus Trotter depth \(L\) at fixed configuration budget \(K=8192\), for naïve sampling and counts-vector proportional stratification, under three measurement models: single-shot (blue \(R=1\)), intermediate (green \(R=64\)), and oracle (red \(R=\infty\)). The horizontal grey line shows the bound \(1/K\), which follows from \(|Y^{(R)}|\le \|g\|_1\) for Pauli observables. Importantly, note how both naive and stratified exhibit the same asymptotic scaling with $\|g\|_1^2$. The corresponding variance ratios \(\rho(R)\) plotted in the main text are obtained by taking, at each \(L\) and \(R\), the ratio of the stratified (solid) and naïve (dashed) curves.

\begin{figure}
    \centering
    \includegraphics[width=0.95\linewidth]{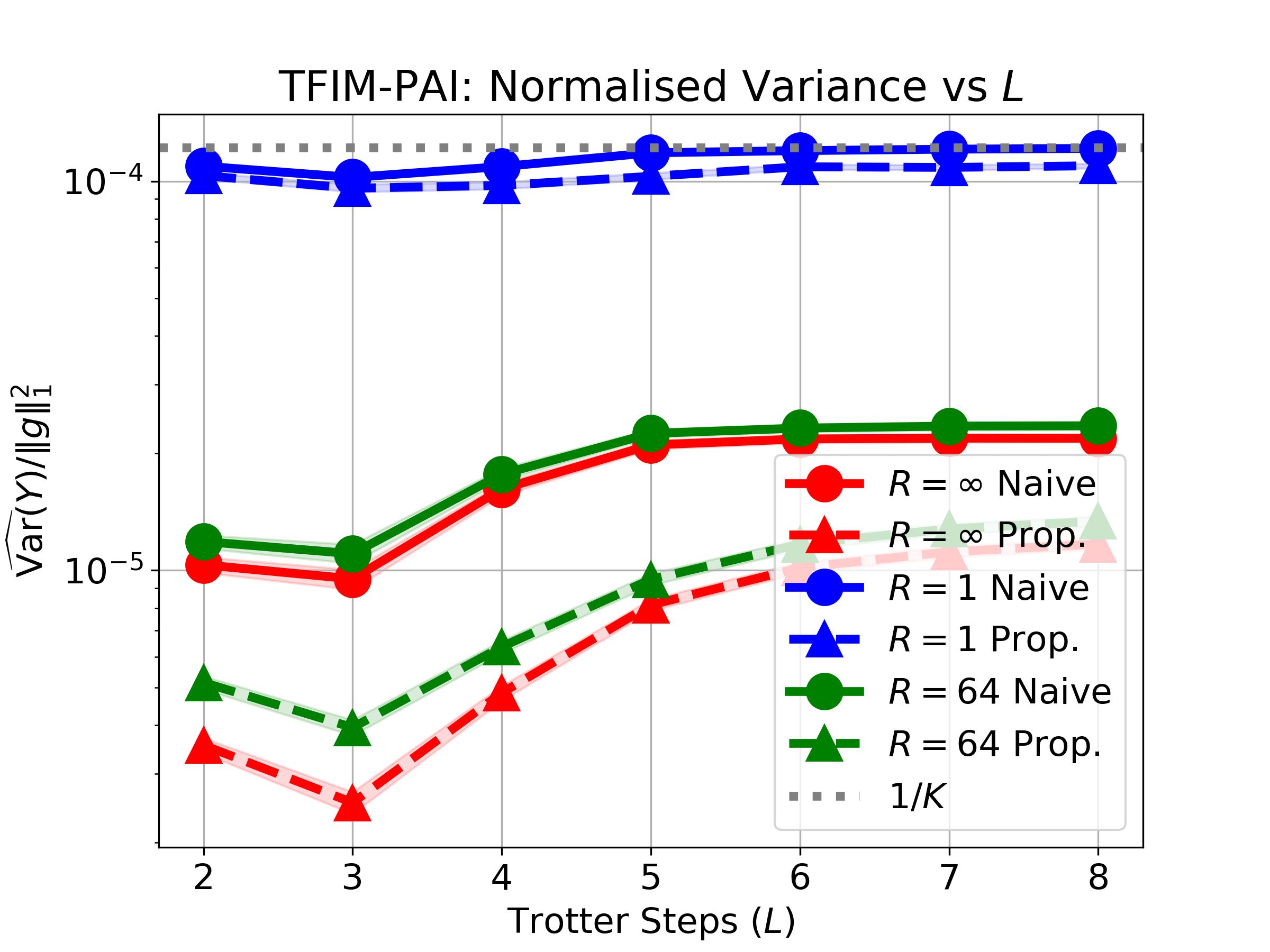}
    \caption{TFIM-PAI absolute variance scaling against 1-norm}
    \label{fig:absvar-PAI}
\end{figure}

\begin{figure}
    \centering
    \includegraphics[width=0.95\linewidth]{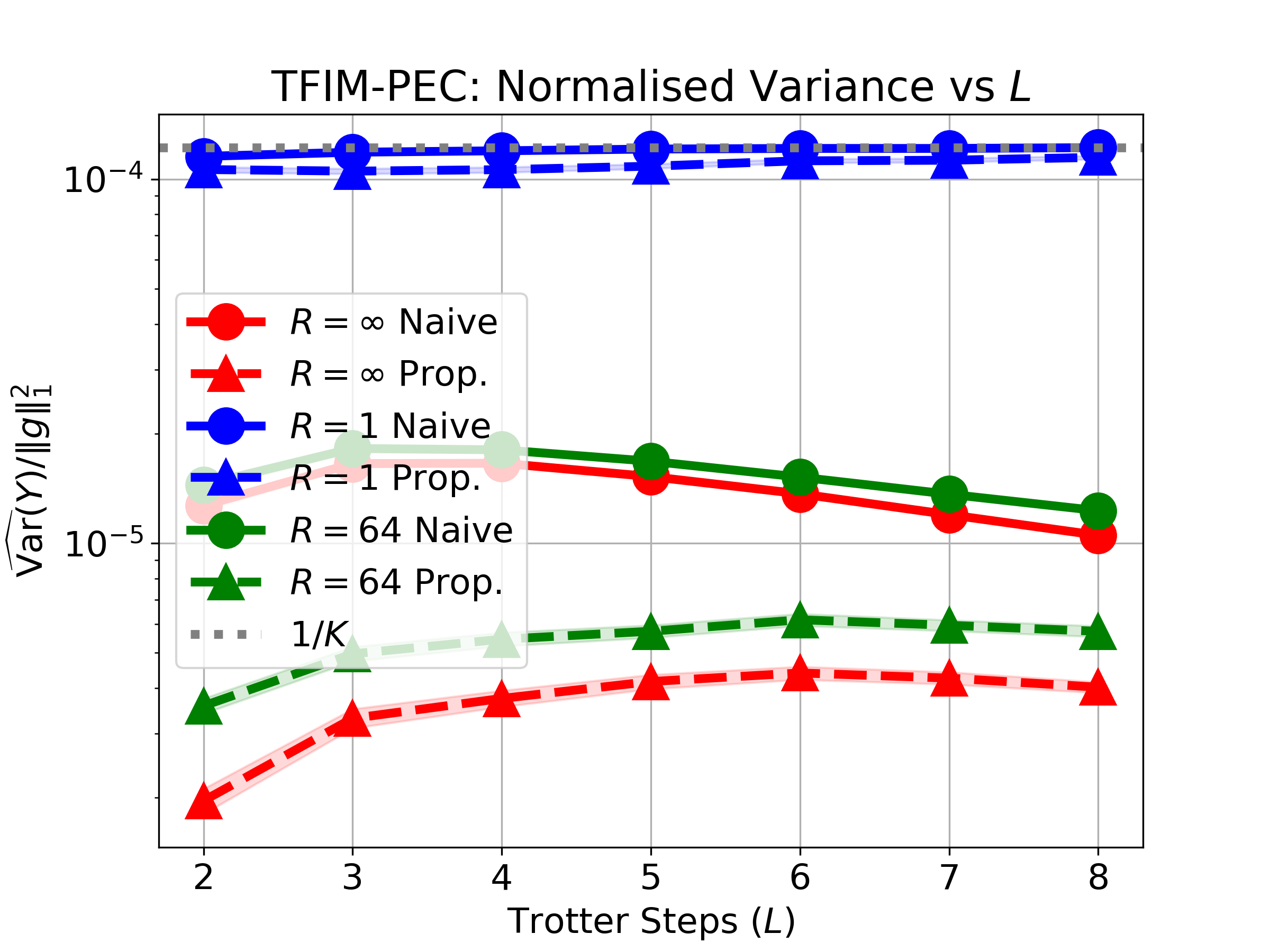}
    \caption{TFIM-PEC absolute variance scaling against 1-norm}
    \label{fig:absvar-pec}
\end{figure}

\subsection{Monte Carlo validation and implementation diagnostics}
\label{subsec:mc-validation}

To validate our implementations (naïve sampling, stratified sampling, and the residual-aware apportionment),
we performed convergence checks on a small PEC instance where exact enumeration is feasible.

\paragraph{Enumerable PEC instance.}
We consider a TFIM Trotter circuit with \(n=3\) qubits, \(L=1\), and open boundary conditions. As in the main text we
fix \(h=0.6\), \(J=0.7\), and \(t=1.0\), and add a single-qubit depolarising noise channel of strength \(p=0.01\)
after each gate (independently for the two legs of two-qubit gates). This yields \(\nu=7\) local PEC QPD locations.
Each location has a 4-term local QPD, so there are \(4^7=16384\) circuit variants. In the oracle model, exact
reference values for the mean and the design variances are available by explicit summation. (When performing classical simulations, Born-rule noise
is just an additional sampling layer on top of oracle outputs, so we focus on the oracle setting here.)

\begin{figure}
    \centering
    \includegraphics[width=0.8\linewidth]{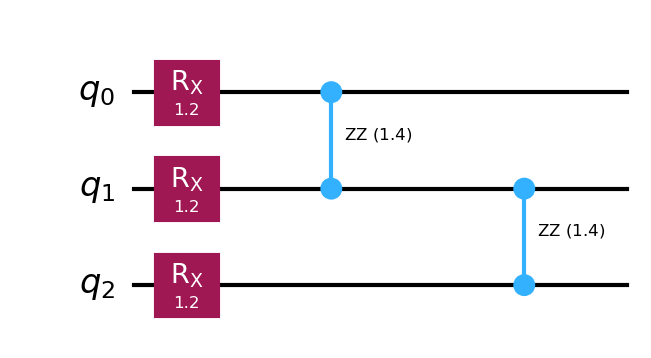}
    \caption{Trotter TFIM circuit with \(n=3\), \(L=1\), open boundary conditions. In the convergence checks we apply a
    single-qubit depolarising channel of strength \(p=0.01\) after each gate (independently for the two-qubit rotations).}
    \label{fig:convergence check circuit}
\end{figure}

\paragraph{Convergence of mean and variance estimates.}
Fig.~\ref{fig:mc-convergence} shows Monte Carlo convergence of (a) the mean
estimator and (b) the scaled estimator variance \(K\,\widehat{\Var}(\hat{\mu})\), for both naïve and stratified
designs, against the exact values obtained by full enumeration. Bands are \(95\%\) bootstrap intervals with
\(B=1024\) resamples, using the procedures in Appendix~\ref{subsec:bootstrap}.

\begin{figure}
    \includegraphics[width=0.95\linewidth]{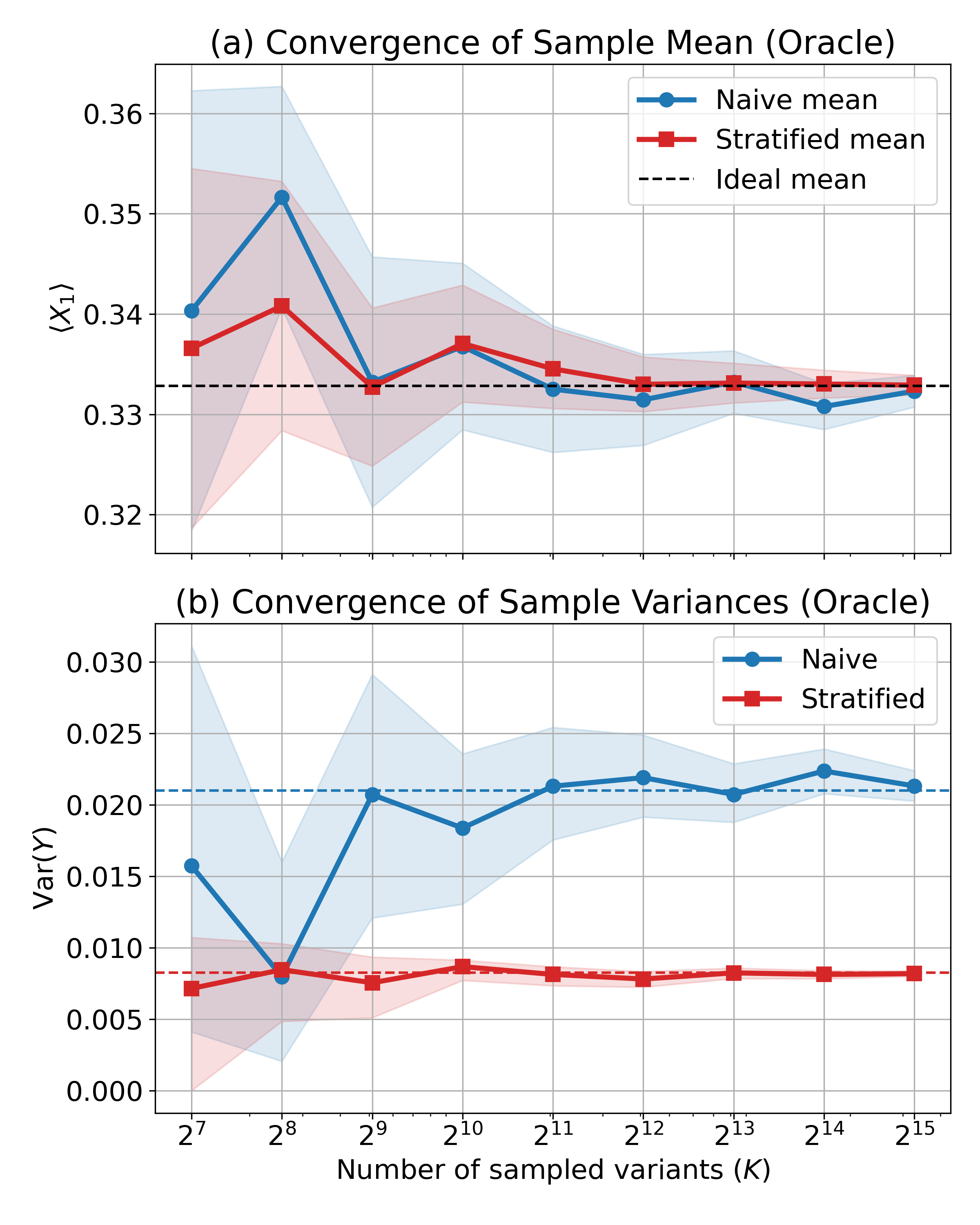}
    \caption{Monte Carlo convergence checks for an exactly enumerable PEC instance in the oracle model. The mean and
    the scaled estimator variances \(K\,\widehat{\Var}(\hat{\mu})\) converge to the exact values obtained by explicit
    summation over all \(4^7\) configurations. Shaded bands denote \(95\%\) bootstrap confidence intervals with
    \(B=1024\) resamples.}
    \label{fig:mc-convergence}
\end{figure}

\subsection{Certificate for apportionment-induced variance perturbations}
\label{subsec:certificate-exact}

The proportional stratification results in Section~\ref{subsec:stratified-qpd-main} are stated for ideal (generally
non-integer) quotas \(K_s=K w_s\). Our implementation uses integer apportionment (Hamilton's method) together with a
residual stratum aggregating all zero-allocation strata, which preserves exact unbiasedness at finite \(K\)
(Appendix~\ref{app:allocation details}). This subsection records a deterministic certificate bounding the deviation
between the implemented estimator variance and the ideal proportional-allocation variance due purely to finite-\(K\)
rounding and residual aggregation.

\paragraph{Deterministic certificate.}
Assume the per-configuration random variable is almost surely bounded as \(|Y|\le B\) (in our PEC/PAI numerics with
Pauli observables, \(B=\ \lVert O\rVert_{\infty}\|g\|_1\) with \(\ \lVert O\rVert_{\infty}=1\)). Let \(K_s\) denote the final integer allocations over the
original strata, and let \((w_*,K_*)\) denote the residual weight and residual shot count after borrowing. Define
\(\mathcal{A}:=\{s:K_s>0\}\) and \(\mathcal{D}:=\{s:K_s=0,\,w_s>0\}\) so that \(w_*=\sum_{s\in\mathcal{D}}w_s\).
Using the perturbation decomposition in Appendix~\ref{app:var-rounding-residual} together with the bounds
\(\sigma_s^2\le B^2\) and \(\Var_{\pi}(\mu_s)\le B^2\), we obtain
\begin{align}
\bigl|\Var(\widehat\mu_{\mathrm{impl}})-\Var(\widehat\mu_{\mathrm{prop}})\bigr|
\le
\mathrm{Cert}_{\mathrm{var}},
\label{eq:cert_var_app_merged}
\end{align}
with 
\[
\mathrm{Cert}_{\mathrm{var}}
:=B^2\left(T_{\mathrm{i}}+T_{\mathrm{ii}}+T_{\mathrm{iii}}\right),
\]
where
\begin{align}
T_{\mathrm{i}} &:= \sum_{s\in\mathcal{A}} w_s^2\left|\frac{1}{K_s}-\frac{1}{K w_s}\right|
\quad\text{(retained rounding)},\\
T_{\mathrm{ii}} &:= w_*\left|\frac{w_*}{K_*}-\frac{1}{K}\right|
\quad\text{(residual quota mismatch)},\\
T_{\mathrm{iii}} &:= \frac{w_*^2}{K_*}
\quad\text{(residual aggregation penalty)}.
\end{align}
Moreover, since \(\bigl|\sqrt{u}-\sqrt{v}\bigr|\le \sqrt{|u-v|}\) for \(u,v\ge 0\), the quantity
\(\sqrt{\mathrm{Cert}_{\mathrm{var}}}\) gives a conservative scale for possible standard-error perturbations induced by
finite-\(K\) apportionment and residual aggregation.

\paragraph{Observed scale.}
On the enumerable PEC instance above, the certificate is conservative by one to three orders of magnitude, consistent
with the use of worst-case bounds. On the main TFIM benchmarks (\(K=8192\)), we observe small residual masses
(\(w_*\sim 10^{-4}\)–\(10^{-3}\) with \(K_*\) of order \(1\)–\(10\)), and \(\sqrt{\mathrm{Cert}_{\mathrm{var}}}\) is a modest
fraction of the empirical standard error. This supports interpreting the reported variance reductions as insensitive to
the finite-\(K\) apportionment details.

\section{Coarsening and hierarchy of permutation-invariant strata}
\label{app:coarsening-main}

Stratification depends on a choice of statistic $S(\underline\ell)$ that partitions the configuration
space into strata. This appendix records a simple but useful monotonicity principle:
\emph{refining} the statistic (i.e.\ using more information) cannot increase the variance of the
\emph{proportional-allocation} stratified estimator, for a fixed measurement model (fixed $R$).

\paragraph{Refinements and monotonicity.}
Let $S$ and $S'$ be two statistics of $\underline\ell$, and say that $S'$ \emph{refines} $S$ if
$S=f(S')$ for some deterministic map $f$ (equivalently, every $S'$-stratum is contained in an
$S$-stratum). Write $\Var(\widehat{Y}_K^{\mathrm{prop}};S)$ for the variance of the proportional
stratified estimator when strata are defined by $S$. Under ideal proportional quotas
$K_s = K\,w_s$ (ignoring integrality), we have
\begin{equation}
\Var(\widehat{Y}_K^{\mathrm{prop}};S')
\;\le\;
\Var(\widehat{Y}_K^{\mathrm{prop}};S).
\label{eq:refinement-monotonicity}
\end{equation}
One way to see this is to note that, under proportional allocation,
\[
\Var(\widehat{Y}_K^{\mathrm{prop}};S)=\frac{1}{K}\,\mathbb E\!\left[\Var(Y\mid S)\right],
\]
and conditioning on a finer $\sigma$-algebra can only decrease conditional variance on average
(i.e.\ $\mathbb E[\Var(Y\mid S')]\le \mathbb E[\Var(Y\mid S)]$ when $S=f(S')$).

\paragraph{The finest statistic and an absolute lower bound.}
The finest possible statistic is the full configuration label
\[
S_{\mathrm{full}}(\underline\ell):=\underline\ell.
\]
In this case each stratum contains a single configuration, so
$\mu_{S_{\mathrm{full}}}=\mu_{\underline\ell}$ and $\sigma^2_{S_{\mathrm{full}}}=\Var(Y\mid\underline\ell)$, and
proportional stratification yields
\begin{align}
\Var(\widehat{Y}_K^{\mathrm{prop}};S_{\mathrm{full}})
&=\frac{1}{K}\sum_{\underline\ell} p(\underline\ell)\,\Var(Y\mid\underline\ell)\nonumber\\
&=\frac{1}{K}\mathbb E_{\underline\ell}\!\left[\Var(Y\mid\underline\ell)\right].
\label{eq:full-stat-lowerbound}
\end{align}
Comparing~\eqref{eq:full-stat-lowerbound} with the law of total variance,
\[
\Var(Y)=\mathbb E_{\underline\ell}[\Var(Y\mid\underline\ell)]
+\Var_{\underline\ell}(\mu_{\underline\ell}),
\]
shows that $S_{\mathrm{full}}$ removes \emph{all}  configuration variance
$\Var_{\underline\ell}(\mu_{\underline\ell})$ and leaves only within-configuration noise (Born-rule
noise under the measurement model). It is therefore an absolute lower bound among stratified
designs that use the same per-configuration estimator $Y$.

\paragraph{Permutation-invariant coarsenings: counts-vector and sign parity.}
In product-form QPDs, $S_{\mathrm{full}}$ has exponentially many strata and consequently calculating the marginal distribution is infeasible. The counts vector $\mathbf M$ introduced in the main text can be viewed as the permutation-invariant coarsening of $S_{\mathrm{full}}$ that remains tractable (polynomially many strata for fixed width $d$) and already guarantees never-worse performance relative to naïve sampling (under ideal proportional quotas). It can be understood as the \textit{finest possible} (and consequently variance-minimal) of the class of statistics that are permutation-invariant.

Further coarsenings of the counts-vector can further reduce preprocessing cost at the price of
more configuration variance. A notable example is \emph{sign parity}, as used in \cite{chen_faster_2025} and very recently in \cite{myers_simulating_2025}.
Partition each local index set into positive and negative coefficient labels,
\[
I_i^+ := \{\ell:\gamma_i(\ell)>0\},\quad
I_i^- := \{\ell:\gamma_i(\ell)<0\},
\]
and define
\[
\mathbf P=(P_+,P_-),\quad
P_\pm := \#\{i:\ell_i\in I_i^\pm\},\quad P_++P_-=\nu.
\]
Equivalently, $\mathbf P$ is obtained from $\mathbf M$ by merging all labels with the same sign, so
each $\mathbf P$-stratum is a union of $\mathbf M$-strata. Hence we have the refinement chain
\[
S_{\mathrm{full}} \;\succ\; \mathbf M \;\succ\; \mathbf P,
\]
and corresponding variance hierarchy (corresponding to successive applications of Theorem~\ref{theorem:proportional variance})
\begin{equation}
\Var(\widehat{Y}_K^{\mathrm{prop}};S_{\mathrm{full}})
\;\le\;
\Var(\widehat{Y}_K^{\mathrm{prop}};\mathbf M)
\;\le\;
\Var(\widehat{Y}_K^{\mathrm{prop}};\mathbf P).
\label{eq:counts-parity-hierarchy}
\end{equation}
The second inequality is tight when the additional multinomial detail in $\mathbf M$ does not change stratum means beyond what is already captured by $\mathbf P$; conversely, it can be strict when different primitive labels (even of the same sign) have systematically different effects. We also note that $\mathbf{P}$ is a universally applicable statistic for all product-form QPDs, since it essentially just partitions based on the amount of `negativity' in a given configuration. However, we choose to focus on the more costly counts-vector scheme since it is the more general construction. 

\paragraph{Cost--variance trade-off.}
The appeal of $\mathbf P$ is that it has only $\nu+1$ values and admits a binomial-style DP with $O(\nu^2)$ time and memory, independent of the original width $d$. A homogeneous version of sign-parity was used in Ref.~\cite{chen_faster_2025}. In symmetric settings (e.g.\ homogeneous single-qubit depolarising PEC), sign parity can capture most of the explainable configuration variance while being far cheaper than the full counts vector. In less symmetric settings, $\mathbf M$ typically yields stronger variance reduction at higher (but still polynomial-in-$\nu$ for fixed $d$) preprocessing cost. More generally, one can interpolate between $\mathbf M$ and $\mathbf P$ by merging subsets of primitive labels whose detailed counts have little impact on the conditional means $\mu_{\mathbf m}$, leading to various polynomial powers in $\nu$ between $\nu^2$ and the full width $\nu ^d$.

\subsection{Counts-vector vs sign-parity stratification on a small PEC instance}
\label{app:counts-vs-parity}

To illustrate this trade-off, we compare counts-vector stratification with sign-parity stratification
on the small PEC instance from Appendix~\ref{app:numerical-methodology}: TFIM Trotter on $n=3$ qubits,
$L=1$ step, open boundaries, and homogeneous single-qubit depolarising noise $p=0.01$ after each gate.
The resulting PEC QPD has $\nu=7$ local factors and $4^7\approx 10^4$ circuit variants.

For the full counts-vector statistic $\mathbf M$ there are
\[
\#\{\mathbf M\}=\binom{\nu+4-1}{4-1}=120
\]
strata, indexed by $\mathbf M=(M_I,M_X,M_Y,M_Z)$ with $M_I+M_X+M_Y+M_Z=7$.
For sign parity $\mathbf P=(P_+,P_-)$ there are only
\[
\#\{\mathbf P\}=\nu+1=8
\]
strata.

Because this instance is small, we compute exact design variances by enumerating all $4^7$ configurations
(oracle model):
\begin{align*}
\Var_{\text{naive}}(Y) &\approx 2.099\times 10^{-2},\\
\Var_{\text{strat}}^{\text{(counts)}}(Y) &\approx 8.246\times 10^{-3},\\
\Var_{\text{strat}}^{\text{(parity)}}(Y) &\approx 8.404\times 10^{-3}.
\end{align*}
Both stratifications reduce configuration variance by $\sim 2.5\times$, and sign parity recovers more than
$98\%$ of the variance reduction achieved by the full counts vector, while using only $8$ strata instead of $120$. 

Fig.~\ref{fig:pec_parity_vs_counts} shows empirical variance estimates converging to these exact baselines as
$K$ increases. This toy example demonstrates that, in specific settings, very coarse statistics can
capture most of the explainable configuration variance at dramatically reduced preprocessing cost. Importantly, the results still verify the hierarchy in~\eqref{eq:counts-parity-hierarchy}. We emphasise that we expect the gap between $\mathbf M$ and $\mathbf P$ to widen in circumstances where the class of negative events is large and when individual primitives within it have drastically different effects. We leave investigating when exactly a coarser vs. finer statistic is optimal (which depends strongly on available resources and problem structure) for future work. 

\begin{figure}
    \centering
    \includegraphics[width=\linewidth]{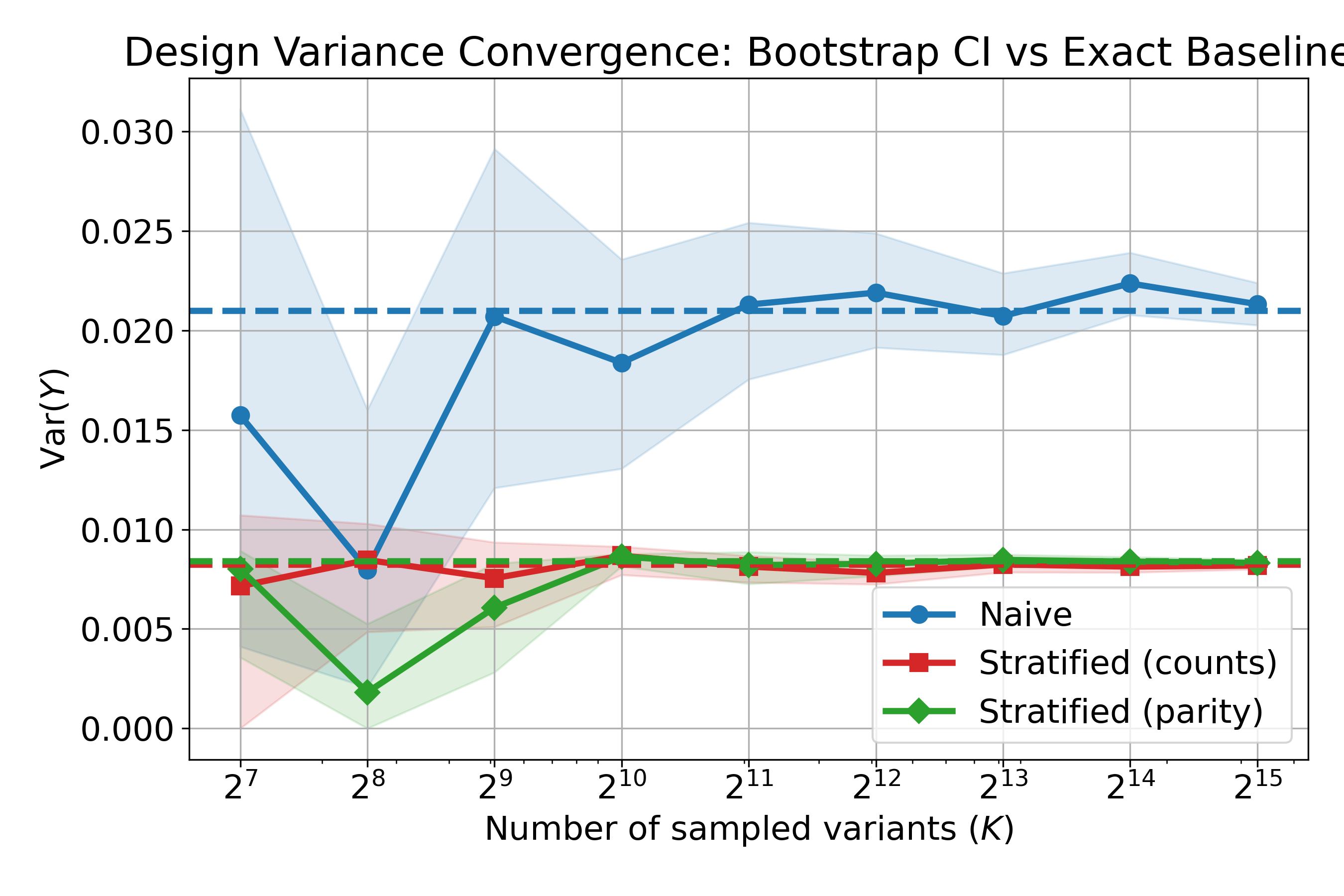}
    \caption{%
    Design-variance convergence for a small PEC instance in the oracle model, comparing naïve sampling (blue) with counts-vector stratification (red) and sign-parity stratification (green). The horizontal dashed lines show the exact design variances obtained by enumerating all $4^7$ PEC configurations: $\Var_{\text{naive}}(Y)\approx 0.02099$, $\Var_{\text{strat}}^{\text{(counts)}}(Y)\approx 0.00825$, and $\Var_{\text{strat}}^{\text{(parity)}}(Y)\approx 0.00840$ (both under proportional allocation). Solid curves show empirical variance estimates as a function of the number of sampled configurations $K$, with shaded bands indicating $95\%$ bootstrap confidence intervals ($M=1024$ resamples). Both stratified schemes converge quickly to their baselines and achieve a $\sim 2.5\times$ reduction relative to naïve sampling, while the coarse sign-parity statistic recovers more than $98\%$ of the variance reduction of the full counts-vector statistic using only $8$ strata instead of $120$.%
    }
    \label{fig:pec_parity_vs_counts}
\end{figure}
\section{Sufficiency, explained variance, and symmetry}
\label{app:sufficiency}

This appendix collects a few formal consequences of the variance-accounting viewpoint used in
Section~\ref{sec:QPDs} and motivates when a stratification statistic $S$ should be expected to help.

\subsection{Explained-variance identity}

Let $Y$ denote the single-configuration scalar random variable used in Monte Carlo estimation
(including QPD reweighting and, in the shot model, Born noise and any $R$-fold averaging),
and let $S=S(\underline\ell)$ be any finite-valued stratification statistic with strata $\mathcal S$.
Write $w_s=\Pr(S=s)$, $\mu_s=\mathbb E[Y\mid S=s]$, and $\sigma_s^2=\Var(Y\mid S=s)$.

\begin{proposition}[Explained-variance interpretation]
\label{prop:explained-variance}
For any statistic $S$,
\[
\Var(Y)
=\mathbb E[\Var(Y\mid S)] + \Var(\mathbb E[Y\mid S]).
\]
Under ideal proportional allocation $K_s=K w_s$,
\[
\Var(\widehat\mu_{\mathrm{prop}})=\frac{1}{K}\sum_{s} w_s\sigma_s^2,
\quad
\Var(\widehat\mu_{\mathrm{naive}})=\frac{1}{K}\Var(Y),
\]
so the design-level variance reduction is
\[
\Var(Y^{\mathrm{naive}})-\Var(Y^{\mathrm{prop}})
=\Var(\mathbb E[Y\mid S]).
\]
\end{proposition}

Defining
\[
R^2_{\mathrm{eff}}(S):=\frac{\Var(\mathbb E[Y\mid S])}{\Var(Y)},
\]
the ideal variance ratio is $\rho(S)=1-R^2_{\mathrm{eff}}(S)$ (as stated in the main text). Thus $S$
is effective precisely when it explains a large fraction of the single-configuration variability in $Y$.

\subsection{Oracle specialisation and configuration-budget savings}

In the oracle model, the per-configuration scalar is deterministic given $\underline\ell$:
\[
Y=\mu_{\underline\ell}:=\mathbb E[Y\mid \underline\ell].
\]
Then $\Var(Y)$ is purely configuration-to-configuration mean variation, and $R^2_{\mathrm{eff}}(S)$ reduces
to the usual explained fraction $R^2(S)$ of $\mu_{\underline\ell}$ by conditioning on $S$.
In particular,
\begin{equation}
\rho^{\mathrm{oracle}}(S)=1-R^2(S).
\label{eq:rho_oracle_R2_app}
\end{equation}

Moreover, because $\Var(\widehat\mu)=\Var(Y)/K$ in both designs, matching a fixed target variance implies
\[
\frac{\Var(Y^{\mathrm{prop}})}{K_{\mathrm{prop}}}=\frac{\Var(Y^{\mathrm{naive}})}{K_{\mathrm{naive}}}
\Longrightarrow
K_{\mathrm{prop}}=\rho^{\mathrm{oracle}}(S)\,K_{\mathrm{naive}},
\]
so in the oracle setting the configuration-budget savings fraction is exactly
$1-\rho^{\mathrm{oracle}}(S)=R^2(S)$.
Outside the oracle model, shot noise increases $\Var(Y)$ without increasing $\Var(\mathbb E[Y\mid S])$,
so $R^2_{\mathrm{eff}}(S)\le R^2(S)$ and the attainable savings are reduced accordingly.

\subsection{Counts vectors, permutation symmetry, and symmetry breaking}

The counts-vector statistic $\mathbf M=S_{\mathrm{count}}(\underline\ell)$ forgets gate ordering and records
only the multiplicities of local primitive indices.
If the oracle means $\mu_{\underline\ell}$ are invariant under a permutation group $\mathbb G$ of gate positions
that preserves $\mathbf M$, then $\mu_{\underline\ell}$ is constant on each counts-vector stratum and
$\Var(\mu_{\underline\ell}\mid \mathbf M)=0$.
In this mean-sufficient limit, $R^2(\mathbf M)=1$ and $\rho^{\mathrm{oracle}}(\mathbf M)=0$, i.e.\ counts-vector
stratification removes all configuration-to-configuration mean variation and only measurement noise remains.

In realistic circuits this symmetry is only approximate (e.g.\ due to non-commutation across layers in $\nu$), so
$\Var(\mu_{\underline\ell}\mid \mathbf M)>0$ and $R^2(\mathbf M)<1$.
Consequently $\rho^{\mathrm{oracle}}(\mathbf M)>0$, and growth of $\rho^{\mathrm{oracle}}$ with circuit depth can be
read as a symmetry-breaking effect: the counts vector explains a smaller fraction of mean variation as depth increases.

\section{Cost models, pilot schemes, and adaptive allocation}
\label{app:cost-pilots}

The main text compares configuration-sampling designs at a fixed measurement model, i.e.\ at a fixed number of Born-rule repetitions $R$ per configuration. In that setting, the relevant single-sample object is the per-configuration average $Y^{(R)}$, and the estimator satisfies
\[
\Var(\widehat{Y}_{K,R})=\frac{1}{K}\Var\!\bigl(Y^{(R)}\bigr).
\]
Consequently, variance ratios between sampling designs translate directly into constant-factor reductions in the number of configurations $K$ required to achieve a target precision.

In hardware experiments, budgets are typically expressed in terms of total cost (e.g.\ wall-clock time), which couples configuration count $K$ and shot count $R$ and motivates joint budgeting and pilot-based variance estimation. A detailed treatment of such hardware-relevant cost models and optimal budgeting between $K$ and $R$ under a joint circuit- and shot-level variance decomposition is given in Ref.~\cite{aharonov_reliable_2025}. Here we only record the basic link between our variance ratios and cost models, and then discuss adaptive (pilot/Neyman-type) allocation
across strata. 

\subsection{Connection to cost models via the variance decomposition}
\label{app:cost-model}

For any fixed configuration-sampling design (naïve or stratified), the law of total variance gives
\begin{equation}
\Var\!\bigl(Y^{(R)}\bigr)
=
\frac{1}{R}\,\mathbb E_{\underline\ell}\!\left[\Var(\tilde Y\mid\underline\ell)\right]
+\Var_{\underline\ell}\!\left(\mu_{\underline\ell}\right),
\label{eq:var-singleconfig-app}
\end{equation}
where $\tilde Y$ denotes a single reweighted shot and $Y^{(R)}$ is the average of $R$ i.i.d.\ shots at fixed $\underline\ell$, with $\mu_{\underline\ell}:=\mathbb E[\tilde Y\mid\underline\ell]$. The first term is Born-rule noise and scales as $1/R$; the second term is configuration variance and is independent of $R$. 

Stratification modifies \eqref{eq:var-singleconfig-app} only through the configuration-variance term. In the ideal proportional model, replacing naïve configuration sampling by stratified sampling reduces the configuration variance by a constant factor $\rho\le 1$, and therefore reduces $\Var(\widehat Y_{K,R})$ by the same constant factor at fixed $(K,R)$. Under any hardware cost model that constrains $(K,R)$ (or allows nonuniform $R_k$), this improvement can be interpreted as a constant-factor reduction in the cost required to reach a fixed target precision, once the remaining budgeting problem is solved. Related analysis of such joint budgeting appears with more hardware error-mitigation focus in \cite{aharonov_reliable_2025}.

\subsection{Pilot and Neyman allocation across strata}
\label{app:pilot-neyman}

For a fixed stratification statistic $S$ and fixed repetitions $R_k\equiv R$, classical sampling theory shows that the variance-minimising allocation across strata is the Neyman rule
\[
K_s^\star
=K\,\frac{w_s\,\sigma_s}{\sum_t w_t\,\sigma_t},
\qquad
\sigma_s^2:=\Var\!\bigl(Y^{(R)}\mid S=s\bigr),
\]
which yields
\[
\Var(\widehat Y_{K,R}^{\min})
=\frac{1}{K}\Bigl(\sum_s w_s\sigma_s\Bigr)^2
\le
\frac{1}{K}\sum_s w_s\sigma_s^2
=
\Var(\widehat Y_{K,R}^{\mathrm{prop}})
\]
by Cauchy--Schwarz. We focus on proportional allocation because it depends only on stratum masses $w_s$, which are available exactly from the DP, whereas Neyman allocation requires (possibly approximate) within-stratum variance
information $\sigma_s$.

A standard approximation is a two-stage pilot scheme: expend a small fraction of the budget to obtain rough within-stratum variance estimates $\widehat\sigma_s^2$, then allocate the remaining configurations according to $K_s\propto w_s\widehat\sigma_s$. Pilot data can be included in the final estimator, so no samples are discarded; the trade-off is that allocation decisions are made using noisy variance estimates. Such pilot-based budgeting is also natural in hardware pipelines, since one typically needs variance estimation even for the naïve protocol (see, e.g., Algorithm~2 in Appendix~B of \cite{aharonov_reliable_2025}). In the stratified setting, the additional granularity is that one estimates $\sigma_s^2$ per stratum rather than a single global variance. We leave a detailed empirical study of pilot/Neyman schemes for QPD stratification to future work.

\section{General channel randomisation viewpoint}
\label{app:general-channel-rand}

The main text develops all variance statements in the concrete setting of product-form QPDs, where circuit variants are labelled by multi-indices $\underline\ell$ and the configuration law $p(\underline\ell)$ factorises over gates. This QPD problem structure is convenient a operational point of view since it admits an efficient DP and sampling procedure, but the mathematics translate readily into more general settings. In this appendix we briefly record how the same variance-accounting and stratification ideas extend to \emph{arbitrary} classical randomisation over quantum channels.

\subsection{From QPD indices to arbitrary channel ensembles}

Let $\rho_0$ be a fixed input state, $\{M_x\}$ a fixed POVM, and $O=\sum_x O_x M_x$ a fixed observable, as in the main text. Consider now an \emph{arbitrary} ensemble of CPTP maps
\[
\Lambda \in \{\lambda\},\quad
\Pr(\Lambda=\lambda)=p(\lambda),
\]
which we interpret as the ``configuration'' random variable. For each realisation $\Lambda=\lambda$, the hybrid protocol:
\begin{itemize}
  \item applies $\lambda$ to $\rho_0$,
  \item measures with $\{M_x\}$, obtaining $X\sim p(x\mid\lambda)$,
  \item maps $X$ to a scalar $Z=O(X)$,
  \item and optionally rescales by a weight $w(\lambda)$ to produce $Y=w(\lambda)\,Z$.
\end{itemize}
In Sec.~\ref{sec:QPDs}, we assumed that $p(\lambda)$ and $w(\lambda)$ were chosen so that $Y$ is unbiased for some target expectation (typically the ideal expectation under a reference channel), but this assumption can be easily dropped to allow for bias or even more general problem settings entirely.

All of the variance formulas used in the main text follow directly by applying the law of total variance to this joint law $p(y,\lambda)=p(y\mid\lambda)p(\lambda)$:
\begin{itemize}
  \item the decomposition into a Born-rule term and a classical channel term;
  \item the shot-accounting formula for $\Var(\widehat Y_{K,R})$ in terms of $K$ configuration draws and $R$ repetitions per configuration (analogue of Eq.~\eqref{eq:var-KR-QPD});
  \item the stratified-variance formula with statistic $S=C(\Lambda)$ and the hierarchy
  \[
  \Var(\widehat{Y}_K^{\min})
  \;\le\;
  \Var(\widehat{Y}_K^{\mathrm{prop}})
  \;\le\;
  \Var(\widehat{Y}_K^{\mathrm{naive}})
  \]
  for Neyman, proportional, and naïve allocation (Sec.~\ref{subsec:stratified-qpd-main}).
\end{itemize}
No step in those derivations uses product structure or QPD-specific properties; only the existence of a discrete channel ensemble and an unbiased scalar estimator is required. The caveat is that to make this variance framework operational in the general setting is non-trivial in general.

In QPD language, the random variable $\Lambda$ is simply relabelled as the multi-index $\underline\ell$, the ensemble $p(\lambda)$ becomes the QPD-induced ensemble $p(\underline\ell)$, and the weight $w(\lambda)$ takes the familiar form $\|g\|_1 \,\mathrm{sign}(g(\underline\ell))$. The QPD setting is therefore a special case of a more general ``channel randomisation'' picture.

\subsection{Stratification beyond product-form ensembles}

From this viewpoint, a stratification statistic is just a measurable map
\[
S = C(\Lambda),
\]
which partitions the channel ensemble $p(\lambda)$ into strata labelled by $s\in\mathcal S$, with marginal probability $w_s=\Pr(S=s)$ and conditional law $p(\lambda\mid S=s)$. Given any such statistic, and a way to:
\begin{enumerate}
  \item evaluate or approximate $w_s$, and
  \item sample $\lambda$ conditionally on $S=s$ (possibly from an approximate distribution)
\end{enumerate}
one can implement proportional (or perhaps pilot-based Neyman) stratified sampling exactly as in Sec.~\ref{sec:QPDs}. The variance guarantees are identical: at fixed measurement model and fixed number of configuration draws $K$, stratified sampling never increases variance compared to naïve sampling, and strictly reduces it whenever the stratum means differ.

The special role of product-form QPDs in the main text is therefore \emph{algorithmic}. For general channel ensembles there is no reason to expect the stratum masses $w_s$ or the conditional laws $p(\lambda\mid S=s)$ to admit efficient exact computation. In generic cases one would need to rely on approximate methods (e.g.\ importance sampling, MCMC over configurations, or heuristic groupings of channels) to implement stratification.

In product-form QPDs, by contrast:
\begin{itemize}
  \item the ensemble is indexed explicitly by multi-indices $\underline\ell$;
  \item the product-distribution $p(\underline\ell)$ factorises over gates;
  \item permutation-invariant statistics such as the counts vector $\mathbf M$ and sign parity $\mathbf P$ are natural and easy to define; and
  \item for these statistics the Poisson–multinomial structure yields an exact dynamic programme which computes both $w_s$ and conditional samplers in time polynomial in the circuit depth $\nu$ (for fixed local width $d$).
\end{itemize}
This is what makes counts-vector (and related permutation-invariant) stratification \emph{practically} attractive in the QPD setting. Additionally, since QPDs experience an exponential variance blow-up with the circuit-depth, variance optimisation is a highly relevant goal in practice. In other randomisation settings variance reduction may not be so important.

\subsection{Other randomisation schemes}

Many hybrid protocols outside the QPD family fit into this general channel-randomisation template: randomised compiling, Pauli or Clifford twirling, randomised product formulae/time evolution, classical shadows and related measurement schemes, and stochastic-phase estimation, among others. In all these cases, one can in principle:
\begin{itemize}
  \item identify the induced channel ensemble $p(\lambda)$ over compiled circuits, time steps, or measurement settings;
  \item choose statistics $S(\Lambda)$ that capture relevant structure (e.g.\ numbers of non-identity twirling gates, light-cone properties, coarse features of the random unitary); and
  \item apply the same stratification logic and variance decompositions as in the QPD case.
\end{itemize}

Working out efficient stratification schemes for these more general settings requires problem-specific analysis, since the clean product structure and Poisson–multinomial DP may no longer be available in general (or may be intractable). Nevertheless, the conceptual message is simple: the ``extra'' configurational variance induced by any channel-randomisation scheme can, in principle, be partially removed by stratified sampling over channels. Variance reduction techniques have been investigated previously outside the QPD setting in e.g. for randomised time evolution \cite{kiss_importance_2023}. 

Product-form QPDs provide a particularly transparent testbed where this idea can be made both rigorous and algorithmically explicit. It will be an interesting line of research to investigate how the ideas presented here can generalise.

\section{PAI and PEC constructions used in the TFIM benchmark}
\label{app:qpd-instances}

This appendix records the two product-form QPD instances used in the TFIM Trotter experiments of
Section~\ref{sec:numerical results}: probabilistic angle interpolation (PAI) and probabilistic error cancellation (PEC).
Both fit the local-QPD framework of Section~\ref{sec:QPDs}: each protocol specifies (i) a local index set \(I_i\),
(ii) implementable CPTP primitives \(\mathcal C_i(\ell_i)\), and (iii) real coefficients \(\gamma_i(\ell_i)\), inducing a
product-form circuit-level ensemble over configuration strings \(\underline\ell\). In both cases we stratify by the
counts vector \(S(\underline\ell)\equiv \mathbf M(\underline\ell)\), with stratum weights and conditional sampling
implemented via the DP of Appendix~\ref{app:counts-dp} and residual-aware allocation via Appendix~\ref{app:allocation details}.

\subsection{Probabilistic angle interpolation (PAI)}
\label{app:qpd-instances:pai}

PAI \cite{koczor_probabilistic_2024} implements a target single-qubit rotation in expectation on hardware that supports
only a discrete set of rotation angles. Each ideal rotation is expressed as a signed combination of three implementable
rotations about the same axis; sampling from this local QPD yields an unbiased estimator whose local variance inflation
is controlled by the corresponding local 1-norm.

\paragraph{Single-gate PAI decomposition.}
Fix a Pauli (or SWAP-type) rotation generated by a Hermitian involution \(G\),
\[
R_G(\theta):=e^{-i\theta G/2},
\quad
G^\dagger=G,
\quad
G^2=\mathbb I,
\]
and write the associated channel as \(\mathcal R_G(\theta)(\rho):=R_G(\theta)\rho R_G(\theta)^\dagger\).
Assume the hardware supports a discrete grid of angles \(\Theta_k:=k\Delta\) with step
\[
\Delta=\frac{2\pi}{2^{B_{\mathrm{PAI}}}},
\]
where \(B_{\mathrm{PAI}}\) is the angular precision (in bits). For a target angle \(\theta\), choose \(k\) such that
\[
\theta=\Theta_k+\vartheta,
\quad
0\le \vartheta<\Delta,
\]
and define three implementable primitives
\[
\mathcal C_1:=\mathcal R_G(\Theta_k),\quad
\mathcal C_2:=\mathcal R_G(\Theta_{k+1}),\quad
\mathcal C_3:=\mathcal R_G(\Theta_k+\pi).
\]
PAI provides coefficients \(\gamma_\ell(\vartheta)\in\mathbb R\) such that the target channel admits the exact local QPD
\begin{equation}
\mathcal R_G(\Theta_k+\vartheta)
=
\gamma_1(\vartheta)\,\mathcal C_1
+\gamma_2(\vartheta)\,\mathcal C_2
+\gamma_3(\vartheta)\,\mathcal C_3 .
\label{eq:pai-local-qpd}
\end{equation}
Define the local 1-norm and categorical sampling distribution
\[
\|\boldsymbol\gamma(\vartheta)\|_1:=\sum_{\ell=1}^3|\gamma_\ell(\vartheta)|,
\quad
p(\ell\mid\vartheta):=\frac{|\gamma_\ell(\vartheta)|}{\|\boldsymbol\gamma(\vartheta)\|_1}.
\]
Sampling follows the standard QPD rule: draw \(\ell\sim p(\cdot\mid\vartheta)\), implement \(\mathcal C_\ell\), and
rescale the measured observable outcome by
\(\|\boldsymbol\gamma(\vartheta)\|_1\,\mathrm{sign}(\gamma_\ell(\vartheta))\).
This yields an unbiased single-gate estimator for the ideal rotation channel, with variance inflation controlled by
\(\|\boldsymbol\gamma(\vartheta)\|_1^2\). Closed-form expressions for \(\gamma(\vartheta)\) are given in
Ref.~\cite{koczor_probabilistic_2024}.

\paragraph{Mapping to the product-form QPD notation.}
Consider a circuit with \(\nu\) parametrised rotations, where gate \(i\) has generator \(G_i\) and target angle \(\theta_i\).
For each \(i\), choose the nearest notch \(\Theta_{k_i}\) and write \(\theta_i=\Theta_{k_i}+\vartheta_i\) with
\(\vartheta_i\in[0,\Delta)\). The local PAI decomposition has uniform width \(d=3\) and matches
Section~\ref{sec:QPDs} by taking:
\begin{itemize}
  \item index set \(I_i=\{1,2,3\}\);
  \item implementable channels
  \(\mathcal C_i(1)=\mathcal R_{G_i}(\Theta_{k_i})\),
  \(\mathcal C_i(2)=\mathcal R_{G_i}(\Theta_{k_i+1})\),
  \(\mathcal C_i(3)=\mathcal R_{G_i}(\Theta_{k_i}+\pi)\);
  \item coefficients \(\gamma_i(\ell):=\gamma_\ell(\vartheta_i)\).
\end{itemize}
The resulting circuit-level QPD has the standard product form
\[
\mathcal U_{\mathrm{circ}}
=
\sum_{\underline\ell\in\{1,2,3\}^\nu}
g(\underline\ell)\,\mathcal U(\underline\ell),
\;
g(\underline\ell)=\prod_{i=1}^\nu \gamma_i(\ell_i),
\;
\|g\|_1=\prod_{i=1}^\nu \|\boldsymbol\gamma_i\|_1,
\]
with configuration distribution \(p(\underline\ell)=|g(\underline\ell)|/\|g\|_1\).

\paragraph{Counts-vector stratification for PAI.}
A PAI configuration is the index string \(\underline\ell\in\{1,2,3\}^\nu\) specifying whether each rotation uses the lower
notch (\(\ell=1\)), upper notch (\(\ell=2\)), or antipodal rotation (\(\ell=3\)). The induced counts vector
\[
\mathbf M=(M_1,M_2,M_3),
\quad
M_\ell:=\#\{i:\ell_i=\ell\},
\]
with $M_1+M_2+M_3=\nu$ forgets gate positions and records only primitive multiplicities. We stratify with \(S(\underline\ell)\equiv \mathbf M\).
The DP in Appendix~\ref{app:counts-dp} computes \(w_{\mathbf m}=\Pr(\mathbf M=\mathbf m)\) for the generally
inhomogeneous product distribution induced by \(\{p_i(\cdot)\}\), and provides an exact sampler for
\(\underline\ell\mid \mathbf M=\mathbf m\). Integer proportional allocations and exact unbiasedness at finite \(K\) use
the residual-aware procedure of Appendix~\ref{app:allocation details}.

\subsection{Probabilistic error cancellation (PEC)}
\label{app:qpd-instances:pec}

This subsection specifies the PEC model used in Section~\ref{sec:numerical results}: a gate-independent single-qubit
depolarising noise model, the inverse-channel QPD, and the induced product-form configuration ensemble.

\paragraph{Gate-independent depolarising noise.}
We assume a single-qubit depolarising channel of strength \(p\) applied after each logical unitary acting on that qubit:
\begin{equation}
\mathcal D_p(\rho)
:=
(1-p)\rho+\frac{p}{3}\bigl(X\rho X+Y\rho Y+Z\rho Z\bigr).
\label{eq:single-qubit-depol}
\end{equation}
In the Pauli-transfer representation, \(\mathcal D_p\) fixes \(I\) and contracts the Bloch components by
\(\lambda_p=1-\tfrac{4p}{3}\).
Given an ideal gate channel \(\mathcal U^{(i)}\) acting on one or two qubits, the corresponding hardware-realised noisy
gate is
\[
\widetilde{\mathcal U}^{(i)}:=\mathcal N^{(i)}\circ \mathcal U^{(i)},
\]
where \(\mathcal N^{(i)}\) is the tensor product of \(\mathcal D_p\) on the qubits in the support of \(\mathcal U^{(i)}\)
(i.e.\ \(\mathcal N^{(i)}=\mathcal D_p\) for single-qubit gates and \(\mathcal N^{(i)}=\mathcal D_p\otimes \mathcal D_p\)
for two-qubit gates).

\paragraph{Inverse-channel QPD for depolarising noise.}
PEC inserts an inverse noise map so that, as linear superoperators, \(\mathcal D_p^{-1}\circ \mathcal D_p\) equals the
identity. Since \(\mathcal D_p\) is diagonal in the Pauli-transfer representation, the inverse linear map satisfies
\(I\mapsto I\) and \(P\mapsto \lambda_p^{-1}P\) for \(P\in\{X,Y,Z\}\).
While \(\mathcal D_p^{-1}\) is not CPTP for \(p>0\), it admits a 4-term QPD over Pauli conjugation channels
\(\mathcal P_P(\rho):=P\rho P\), namely
\begin{equation}
\mathcal D_p^{-1}
=
\sum_{P\in\{I,X,Y,Z\}}
\gamma(P)\,\mathcal P_P,
\label{eq:depol-inverse-qpd}
\end{equation}
with coefficients
\[
\gamma(I)=\frac{\lambda_p+3}{4\lambda_p},
\quad
\gamma(X)=\gamma(Y)=\gamma(Z)=\frac{\lambda_p-1}{4\lambda_p}.
\]
We define the local 1-norm and sampling probabilities
\[
\|\boldsymbol\gamma\|_1:=\sum_P|\gamma(P)|,
\quad
p(P):=\frac{|\gamma(P)|}{\|\boldsymbol\gamma\|_1},
\]
and signed CPTP primitives \(\mathcal C(P):=\mathrm{sign}(\gamma(P))\,\mathcal P_P\), so that
\eqref{eq:depol-inverse-qpd} matches the standard local-QPD form~\eqref{eq:local-QPD-main} with uniform width \(d=4\).
For a two-qubit gate with independent depolarising noise on each leg, we take the tensor-product inverse
\(\mathcal D_p^{-1}\otimes \mathcal D_p^{-1}\), i.e.\ each noisy \emph{leg} contributes one independent 4-term
single-qubit QPD over \(\{I,X,Y,Z\}\).

\paragraph{Circuit-level product QPD for PEC.}
Let \(\nu\) denote the total number of single-qubit noise locations (“noisy legs”) after expanding each two-qubit gate
into two independent legs. At each location \(i\in[\nu]\) we have the same local QPD
\[
\mathcal N^{(i)}{}^{-1}
=
\sum_{\ell_i\in\{I,X,Y,Z\}}
\gamma(\ell_i)\,\mathcal C(\ell_i),
\]
acting on the corresponding qubit. Expanding the product over \(\nu\) locations yields the product-form circuit-level QPD
\[
\mathcal U_{\mathrm{circ}}
=
\sum_{\underline\ell\in\{I,X,Y,Z\}^\nu}
g(\underline\ell)\,\mathcal U(\underline\ell),
\quad
g(\underline\ell)=\prod_{i=1}^{\nu}\gamma(\ell_i),
\]
where \(\underline\ell=(\ell_1,\dots,\ell_\nu)\) records which Pauli conjugations are applied by the inverse-noise gadgets
and \(\mathcal U(\underline\ell)\) is the resulting “dressed” circuit.

\paragraph{Counts-vector stratification for PEC.}
Under the induced product sampling distribution \(p(\underline\ell)\propto |g(\underline\ell)|\), a configuration is a
Pauli label string \(\underline\ell\in\{I,X,Y,Z\}^\nu\). The associated counts vector
\[
\mathbf M=(M_I,M_X,M_Y,M_Z),
\quad
\sum_{P\in\{I,X,Y,Z\}} M_P=\nu,
\]
records how often each Pauli channel is selected across noisy legs, independent of position. We stratify with
\(S(\underline\ell)\equiv \mathbf M\). Relative to naïve PEC sampling, the only operational change is how
\(\underline\ell\) is chosen: we allocate a deterministic number of draws to each \(\mathbf m\) (via the residual-aware
integer proportional allocation of Appendix~\ref{app:allocation details}), and then sample \(\underline\ell\) conditionally
within each stratum using Appendix~\ref{app:counts-dp}. All subsequent steps (running the dressed circuit and reweighting
outcomes by the standard PEC weight \(w(\underline\ell)=\|g\|_1\,\mathrm{sign}(g(\underline\ell))\)) are unchanged.

\end{document}